\documentclass{ws-rmp}

\usepackage[utf8]{inputenc}
\usepackage[T1]{fontenc}
\usepackage{color}

\usepackage{amsmath}
\usepackage{amsfonts}
\usepackage{amssymb}
\usepackage{amsopn}

\usepackage{units}
\usepackage[all,cmtip]{xy}
\usepackage{diagxy}

\providecommand{\R}{\mathbb{R}}

\providecommand{\C}{\mathbb{C}}
\renewcommand{\C}{\mathbb{C}}

\providecommand{\T}{\mathbb{T}}
\providecommand{\N}{\mathbb{N}}
\providecommand{\Z}{\mathbb{Z}}
\providecommand{\eps}{\varepsilon}

\providecommand{\abs}[1]{\left \lvert #1 \right \rvert}
\providecommand{\sabs}[1]{\lvert #1 \vert}
\providecommand{\babs}[1]{\bigl \lvert #1 \bigr \rvert}

\providecommand{\norm}[1]{\left \lVert #1 \right \rVert}
\providecommand{\snorm}[1]{\lVert #1 \rVert}
\providecommand{\bnorm}[1]{\bigl \lVert #1 \bigr \rVert}
\providecommand{\Bnorm}[1]{\Bigl \lVert #1 \Bigr \rVert}
\providecommand{\scpro}[2]{\left \langle #1 , #2 \right \rangle}

\providecommand{\bscpro}[2]{\bigl \langle #1 , #2 \bigr \rangle}

\providecommand{\sopro}[2]{\vert #1 \rangle \langle #2 \vert}

\providecommand{\dd}{\mathrm{d}}
\providecommand{\id}{\mathrm{id}}

\renewcommand{\Re}{\mathrm{Re} \,}

\providecommand{\ie}{i.~e.~}
\providecommand{\eg}{e.~g.~}

\providecommand{\Hil}{\mathcal{H}}
\providecommand{\Cont}{\mathcal{C}}
\providecommand{\BCont}{\mathcal{C}_{\mathrm{b}}}
\providecommand{\Weyl}{\sharp}

\providecommand{\piref}{\pi_{\mathrm{ref}}}

\providecommand{\Hper}{H_{\mathrm{per}}}

\providecommand{\Zak}{\mathcal{Z}}
\providecommand{\BZ}{M^*}
\providecommand{\WS}{M}
\providecommand{\heff}{{h_{\mathrm{eff}}}}
\providecommand{\El}{{\mathsf{E}}}

\providecommand{\Hoer}[1]{\mathcal{S}^{#1}}

\providecommand{\PSpace}{\Xi}

\providecommand{\ordere}[1]{\mathcal{O}(\eps^{#1})}

\providecommand{\heff}{h_{\mathrm{eff}}}

\providecommand{\heff}{{h_{\mathrm{eff}}}}

\providecommand{\Href}{\mathcal{H}_{\mathrm{ref}}}

\providecommand{\PA}{\mathsf{P}^A}
\providecommand{\Qe}{\mathsf{Q}}

\providecommand{\WeylSys}{W}

\providecommand{\Htau}{\Hil_{\tau}}
\providecommand{\Schwartz}{\mathcal{S}}
\providecommand{\magW}{\Weyl^B}
\providecommand{\magWel}{\magW_{\eps,\lambda}}
\providecommand{\WS}{M}
\providecommand{\Hper}{H_{\mathrm{per}}}

\providecommand{\Weyl}{\mathrm{W}^M_{\varepsilon}}

\providecommand{\Href}{\mathcal{H}_{\mathrm{ref}}}

\providecommand{\Schwartz}{\mathcal{S}}
\providecommand{\Piref}{\Pi_{\mathrm{ref}}}

\providecommand{\Zak}{\mathcal{Z}}
\providecommand{\Hfast}{\mathcal{H}_{\mathrm{fast}}}
\providecommand{\Hslow}{\mathcal{H}_{\mathrm{slow}}}
\providecommand{\Hper}{H_{\mathrm{per}}} 

\providecommand{\piref}{\pi_{\mathrm{ref}}}
\providecommand{\Index}{\mathcal{I}}
\providecommand{\Htau}{L^2_{\tau}(M^{\ast},L^2(M))}

\providecommand{\Fs}{\mathcal{F}_{\sigma}}

\providecommand{\Hslow}{\mathcal{H}_{\mathrm{slow}}}
\providecommand{\Hfast}{\mathcal{H}_{\mathrm{fast}}}

\providecommand{\Qe}{{Q_{\eps}}}

\providecommand{\Hoer}[1]{\mathcal{S}^{#1}}
\providecommand{\Hoerr}[2]{\mathcal{S}^{#1}_{#2}}
\providecommand{\Hoermr}[2]{\mathcal{S}^{#1}_{#2}}

\providecommand{\SemiHoer}[1]{\mathrm{A}\Hoer{#1}}

\providecommand{\SemiHoermr}[2]{\mathrm{A}\mathcal{S}^{#1}_{#2}}

\providecommand{\heff}{h_{\mathrm{eff}}}
\providecommand{\reff}{{r_{\mathrm{eff}}}}
\providecommand{\keff}{{k_{\mathrm{eff}}}}
\providecommand{\BerryC}{\mathcal{A}}

\providecommand{\Eb}{E_{\ast}}

\providecommand{\KA}{{\mathsf{K}^A}}
\providecommand{\Reps}{{\mathsf{R}}}

\providecommand{\order}{\mathcal{O}}
\providecommand{\ordern}{\mathcal{O}_{\norm{\cdot}}}
\providecommand{\BZak}{BFZ~}
\providecommand{\Opx}{\mathrm{Op}}
\providecommand{\Opk}{\mathfrak{Op}}
\providecommand{\SemiTau}{A \Cont^{\infty}_{\tau}}

\providecommand{\opHamiltonian}{\hat{H}}
\providecommand{\Hamiltonian}{H}

\usepackage{enumerate}
\usepackage{tabularx}

\begin{document}

\markboth{G.~De~Nittis and M.~Lein}
{Applications of Magnetic $\Psi$DO Techniques to SAPT}

%
\catchline{}{}{}{}{}
%

\title{Applications of Magnetic $\Psi$DO \linebreak Techniques to SAPT}

\author{Giuseppe De Nittis}

\address{SISSA, via Bonomea, 265
34136 Trieste TS, Italy \\
\email{denittis@sissa.it}}

\author{Max Lein}

\address{Technische Universität München, Zentrum Mathematik, Boltzmannstraße 3, 85747 Garching, Germany\\
\email{lein@ma.tum.de}}

\maketitle

\begin{history}
\received{(14 September 2010)}
\revised{(14 September 2010)}
\end{history}

\begin{abstract}
	In this review, we show how advances in the theory of magnetic pseudodifferential operators (magnetic~$\Psi$DO) can be put to good use in space-adiabatic perturbation theory (SAPT). As a particular example, we extend results of \cite{PST:effDynamics:2003} to a more general class of magnetic fields: we consider a single particle moving in a periodic potential which is subjectd to a weak and slowly-varying electromagnetic field. In addition to the semiclassical parameter $\eps \ll 1$ which quantifies the separation of spatial scales, we explore the influence of an additional parameter $\lambda$ that allows us to selectively switch off the magnetic field. 
	
	We find that even in the case of magnetic fields with components in $\BCont^{\infty}(\R^d)$, \eg for constant magnetic fields, the results of Panati, Spohn and Teufel hold, \ie to each isolated family of Bloch bands, there exists an associated almost invariant subspace of $L^2(\R^d)$ and an effective hamiltonian which generates the dynamics within this almost invariant subspace. In case of an isolated non-degenerate Bloch band, the full quantum dynamics can be approximated by the hamiltonian flow associated to the semiclassical equations of motion found in \cite{PST:effDynamics:2003}.
\end{abstract}



\section{Introduction} 
\label{intro}
A fundamental and well-studied problem is that of a Bloch electron subjected to an electric field and a constant magnetic field where the dynamics is generated by 
\begin{align}
	\hat{H} \equiv \hat{H}(\eps,\lambda) := \tfrac{1}{2} \bigl ( -i \nabla_x - \lambda A(\eps x) \bigr )^2 + V_\Gamma(x) + \phi(\eps x) 
	\label{intro:eqn:full_hamiltonian}
\end{align}
acting on $L^2(\R^d_x)$. Here, $\eps$ and $\lambda$ are dimensionless nonnegative parameters whose significance will be discussed momentarily.  $A$ is assumed to be \emph{a} smooth polynomially bounded vector potential to the magnetic field $B = \dd A$ which is constant in time and uniform in space. Since the magnetic field is uniform, $A$ needs to grow at least linearly. The potential generated by the nuclei and all other electrons $V_{\Gamma}$ is periodic with respect to the crystal lattice \cite{Cances_Deleurence_Lewin:modelling_local_defects_crystals:2008,Cances_Deleurence_Lewin:non-perurbative_defects_crystal:2008}
\begin{align}
	\Gamma := \Bigl \{ \gamma \in \R^d \; \vert \; \mbox{$\gamma = \sum_{j = 1}^d \alpha_j e_j$} , \; \alpha_j \in \Z \Bigr \}
	\label{intro:eqn:lattice}
\end{align}
where the family of vectors $\{ e_1 , \ldots , e_d \}$ which defines the lattice forms a basis of $\R^d$. The potential assumed to be infinitesimally bounded with respect to $- \tfrac{1}{2} \Delta_x$. By Theorem~XIII.96 in \cite{Reed_Simon:M_cap_Phi_4:1978}, this is ensured by the following 
\begin{assumption}[Periodic potential]\label{intro:assumption:V_Gamma}
	We assume that $V_{\Gamma}$ is $\Gamma$-periodic, \ie $V_{\Gamma}(\cdot + \gamma) = V_{\Gamma}(\cdot)$ for all $\gamma \in \Gamma$, and $\int_{\WS} \dd y \, \abs{V_{\Gamma}(y)} < \infty$. 
\end{assumption}
The dual lattice $\Gamma^*$ is spanned by the \emph{dual basis} $\{e_1^\ast,\ldots,e_d^\ast \}$, \ie the set of vectors which satisfy $e_j \cdot e_k^{\ast} = 2 \pi \delta_{kj}$. The assumption on $V_{\Gamma}$ ensures the unperturbed periodic hamiltonian 
\begin{align}
	\hat{H}_{\mathrm{per}} = \tfrac{1}{2} (- i \nabla_x)^2 + V_{\Gamma}
	\label{intro:eqn:per_ham}
\end{align}
defines a selfadjoint operator on the second Sobolev space $H^2(\R^d)$ and gives rise to Bloch bands in the usual manner (cf.~Section~\ref{rewriting:Zak}): the unitary Bloch-Floquet-Zak transform $\Zak$ defined by equation~\eqref{rewrite:Zak:eqn:Zak_transform} decomposes $\hat{H}_{\mathrm{per}}$ into the \emph{fibered} operator 
\begin{align*}
	 \hat{H}_{\mathrm{per}}^{\Zak} := \Zak \hat{H}_{\mathrm{per}} \Zak^{-1} = \int_{\BZ}^{\oplus} \dd k \, \Hper^{\Zak}(k) := \int_{\BZ}^{\oplus} \dd k \, \Bigl ( \tfrac{1}{2} \bigl ( - i \nabla_y + k \bigr )^2 + V_{\Gamma}(y) \Bigr )
\end{align*}
where we have introduced the Brillouin zone 
\begin{align}
	\BZ := \Bigl \{ k \in \R^d \; \vert \; \mbox{$k = \sum_{j = 1}^d \alpha_j e^\ast_j$} , \; \alpha_j \in [-\nicefrac{1}{2},+\nicefrac{1}{2}] \Bigr \}
	\label{intro:eqn:BZ} 
\end{align}
as fundamental cell in reciprocal space. For each $k \in \BZ$, the eigenvalue equation 
\begin{align*}
	\Hper(k) \varphi_n(k) = E_n(k) \, \varphi_n(k) 
	, 
	&& 
	\varphi_n(k) \in L^2(\T^d_y) 
	, 
\end{align*}
where $\T^d_y := \R^d / \Gamma$, is solved by the Bloch function associated to the $n$th band. Assume for simplicity we are given a band $\Eb$ which does not intersect or merge with other bands (\ie there is a \emph{local} gap in the sense of Assumption~\ref{mag_sapt:mag_wc:defn:gap_condition}). Then common lore is that transitions to other bands are exponentially suppressed and the effective dynamics for an initial state localized in the eigenspace associated to $\Eb$ is generated by $\Eb(- i \nabla_x)$ \cite{Grosso_Parravicini:solid_state_physics:2003,Ashcroft_Mermin:solid_state_physics:2001}. 
\medskip

\noindent
If we switch on a constant magnetic field, no matter how weak, the Bloch bands are gone as there is no \BZak decomposition with respect to $\Gamma$ for hamiltonian~\eqref{intro:eqn:full_hamiltonian}. As a matter of fact, the spectrum of $\hat{H}$ is a Cantor set \cite{Gruber:noncommutative_Bloch:2001} if the flux through the Wigner-Seitz cell 
\begin{align}
	\WS := \Bigl \{ y \in \R^d \; \vert \; \mbox{$y = \sum_{j = 1}^d \alpha_j e_j$} , \; \alpha_j \in [-\nicefrac{1}{2},+\nicefrac{1}{2}] \Bigr \}
	\label{intro:eqn:WS} 
\end{align}
is irrational. Even if the flux through the unit cell is rational, we recover only \emph{magnetic} Bloch bands that are associated to a larger lattice $\Gamma' \supset \Gamma$. A natural question is if it is at all possible to see signatures of \emph{nonmagnetic} Bloch bands if the applied magnetic field is weak? 

Our main result, Theorem~\ref{effective_quantum_dynamics:effective_dynamics:thm:adiabatic_decoupling}, answers this question in the positive in the following sense: if the electromagnetic field varies on the macroscopic level, \ie $\eps \ll 1$, then to leading order the dynamics is still generated by the so called \emph{Peierls substitution} $\Eb \bigl ( -i \nabla_x - \lambda A(\eps x) \bigr ) + \phi \bigl (\eps x \bigr )$ (defined as a \emph{magnetic pseudodifferential operator} through equation~\eqref{rewrite:magnetic_weyl_calculus:eqn:mag_wq_torus}, cf.~Section~\ref{rewriting:magnetic_weyl_calculus}). Hence, the dynamics are dominated by the Bloch bands even in the presence of a weak, but \emph{constant} magnetic field. Furthermore, we can derive corrections to any order in $\eps$ in terms of Bloch bands, Bloch functions, the magnetic \emph{field} and the electric \emph{potential}. We do not need to choose a ``nice'' vector potential for $B$, in fact, in all of the calculations only the magnetic field $B$ enters. Existing theory is ill-equipped to deal with constant or even more general magnetic fields. We tackle this obstacle by incorporating the rich theory of magnetic Weyl calculus \cite{Mantoiu_Purice:magnetic_Weyl_calculus:2004,Iftimie_Mantiou_Purice:magneticPseudodifferentialOperators:2005} with \emph{two} small parameters \cite{Lein:two_parameter_asymptotics:2008} into space-adiabatic perturbation theory \cite{PST:sapt:2002,PST:effDynamics:2003,Teufel:adiabaticPerturbationTheory:2003}. 

Magnetic Weyl calculus does not single out constant magnetic fields, in fact we will only assume that the components of the magnetic field $B$ are bounded with bounded derivatives to any order. 
\begin{assumption}[Electromagnetic fields]\label{intro:assumption:em_fields}
	We assume that the components of the external (macroscopic) magnetic fields $B$ and the electric potential $\phi$ are $\BCont^{\infty}(\R^d)$ functions, \ie smooth, bounded functions with bounded derivatives to any order. 
\end{assumption}

\begin{remark}
	All vector potentials $A$ associated to magnetic fields $B = \dd A$ with components in $\BCont^{\infty}(\R^d)$ are always assumed to have components in $\Cont^{\infty}_{\mathrm{pol}}(\R^d)$. This is always possible as one could pick the transversal gauge, 
	\begin{align*}
		A_k(x) := - \sum_{j = 1}^n \int_0^1 \dd s \, B_{kj}(s x) \, s x_j 
		. 
	\end{align*}
	This is to be contrasted with the original work of Panati, Spohn and Teufel where the \emph{vector potential} had to have components in $\BCont^{\infty}(\R^d)$. 
\end{remark}
Under Assumptions~\ref{intro:assumption:V_Gamma} and~\ref{intro:assumption:em_fields}, $\hat{H}$ defines an essentially selfadjoint operator on $\Cont^{\infty}_0(\R^d_x) \subset L^2(\R^d_x)$. 

The extension of the rigorous derivation of the Peierls substitution to the case of constant magnetic field is not only an ``accademic  result,'' it is crucial for modeling the Quantum Hall effect (QHE): the Peierls substitution can be used to link the QHE to the well studied Harper equation \cite{Harper:single_band:1955,Bellissard:Cstar_algebras_solid_state_physics:1988,Helffer_Sjoestrand:semiclassics_Harper_equation_1:1988,Helffer_Sjoestrand:semiclassics_Harper_equation_2:1990,Helffer_Sjoestrand:semiclassics_Harper_equation_3:1989}. 
In particular  the Peierls substitution led Hofstadter to study a simplified tight binding model for the QHE today called the Hofstadter model \cite{Hofstadter:Bloch_electron_magnetic_fields:1976}. The Hofstadter model is a paradigm in the study of fractal spectra (Hofstadter butterfly) and was used by Thouless et al{.} in the seminal paper \cite{Thouless_Kohmoto_Nightingale_Den_Nijs:quantized_hall_conductance:1982} to give the first theoretical explanation of the topological quantization of the QHE. 
More recently Avron \cite{Avron_Osadchy:Hofstadter_butterfly:2001} interpreted the results by Thouless et al{.} from the viewpoint of thermodynamics and connected the QHE to anomalous phase transition diagrams (colored quantum butterflies). Even though this list of publications is very much incomplete, it shows the significance of the Peierls substitution in the study of the QHE. The main merit of this paper is to provide the first rigorous proof of the Peierls substitution under the conditions relevant to the QHE and consequently the first rigorous justification for the use of the Hofstadter model as weak field limit for the analysis of the QHE.

Let us now explain why we have chosen to include the additional parameter $\lambda$ in the hamiltonian $\hat{H} \equiv \hat{H}(\eps,\lambda)$. Our proposal is to model an experimental setup that applies an external, \ie macroscopic electric and magnetic field. The parameter $\eps \ll 1$ relates the microscopic scale as given by the crystal lattice to the scale on which the external fields vary. \emph{We always assume $\eps$ to be small.} It is quite easy to fathom an apparaturs where electric and magnetic field can be regulated separately by, say, two dials. We are interested in the case where we can \emph{selectively switch off the magnetic field}. If we regulate the strength of the magnetic field by varying the relative amplitude $\lambda \leq 1 $ which quantifies the ratio between scaled electric and magnetic field, 
\begin{align*}
	B^{\eps,\lambda}(x) := \eps \lambda B(\eps x),\qquad\qquad 
	\El^{\eps}(x) := \eps \El(\eps x) 
	. 
\end{align*}
We emphasize that $\lambda$ need not be small, as a matter of fact, $\lambda = 1$ is perfectly admissible. In this sense, $\lambda$ is a perturbative parameter which allows us to take the limit $B^{\eps,\lambda} \rightarrow 0$ without changing the external electric field $\El^{\eps}$. 
This is very much relevant to experiments since magnetic fields are typically much weaker than electric fields and thus both, from a physics and a mathematics perspective, the study of the dynamics under the $\lambda \rightarrow 0$ limit is an interesting problem which merits further research. 
\medskip

\noindent
The aim of this review is to show how recent advances in the theory of \emph{magnetic pseudo\-differential operators} (magnetic $\Psi$DOs) (see~\cite{Iftimie_Mantiou_Purice:magneticPseudodifferentialOperators:2005,Lein:two_parameter_asymptotics:2008} as well as section~\ref{rewriting:magnetic_weyl_calculus}) can be used to extend the range of validity of the results of Panati, Spohn and Teufel derived via space-adiabatic perturbation theory (SAPT) \cite{PST:effDynamics:2003} to magnetic fields with components in $\BCont^{\infty}(\R^d)$. The original proof uses standard pseudodifferential techniques and thus is limited to magnetic \emph{vector potentials} of class $\BCont^{\infty}(\R^d)$. In a recent work by one of the authors with Panati \cite{DeNittis_Panati:Bloch_electron:2010}, adiabatic decoupling for the Bloch electron has been proven in the case of constant magnetic field. Their proof rests on a particular choice of gauge, namely the symmetric gauge. We take a different route: according to the philosophy of magnetic Weyl calculus, it is the properties of the magnetic field and not those of the vector potentials which enter the hypotheses of theorems. 

As the proofs in \cite{PST:effDynamics:2003} carry over \emph{mutatis mutandis}, we feel it is more appropriate to elucidate the structure of the problem and mention the necessary modifications in proofs when necessary. 
\medskip

\noindent
Our paper is divided in 5 sections: in \textbf{Section~\ref{rewriting}}, we will decompose the hamiltonian using the Bloch-Floquet-Zak transform and rewrite it as \emph{magnetic} Weyl quantization of an operator-valued function. \textbf{Section~\ref{mag_sapt}} contains a comprehensive description of our technique of choice, SAPT. The main results, adiabatic decoupling to all orders and a semiclassical limit, will be stated and proven in \textbf{Section~\ref{effective_quantum_dynamics}} and \textbf{\ref{eom}}.

\paragraph{Acknowledgements} 
\label{intro:acknowledgements}
The authors would like to thank M.~Măntoiu, G.~Panati and H.~Spohn for useful discussions. M.~L. thanks G.~Panati for initiating the scientific collaboration with M.~Măntoiu. Furthermore, the authors have found the suggestions by one of the referees very useful. 



\section{Rewriting the problem} 
\label{rewriting}
As a preliminary step, we will rewrite the problem: first, we extract the Bloch band picture via the \BZak transform and then we reinterpret the BFZ-transformed hamiltonian as magnetic quantization of an operator-valued symbol. We insist we only rephrase the problem, \emph{no additional assumptions} are introduced.

\subsection{The Bloch-Floquet-Zak transform} 
\label{rewriting:Zak}
Usually, one would exploit lattice periodicity by going to the Fourier basis: each $\Psi \in \Schwartz(\R^d_x) \subset L^2(\R^d_x)$ is mapped onto 
\begin{align*}
	( \mathcal{F} \Psi )(k,y) := \sum_{\gamma \in \Gamma} e^{- i k \cdot y} \, \Psi(y + \gamma) 
\end{align*}
and the corresponding representation is usually called Bloch-Floquet representation (see \eg \cite{Kuchment:Floquet_theory:1993}). It is easily checked that 
\begin{align*}
	( \mathcal{F} \Psi )(k - \gamma^*,y) &= ( \mathcal{F} \Psi )(k,y) 
	&& \forall \gamma^* \in \Gamma^* 
	\\
	( \mathcal{F} \Psi )(k,y - \gamma) &= e^{- i k \cdot \gamma} ( \mathcal{F} \Psi )(k,y) 
	&& \forall \gamma \in \Gamma 
\end{align*}
holds and $\mathcal{F} \Psi$ can be written as 
\begin{align*}
	( \mathcal{F} \Psi )(k,y) = e^{i k \cdot y} \, u(k,y)
\end{align*}
where $u(k,y)$ is $\Gamma$-periodic in $y$ and $\Gamma^*$-periodic up to a phase in $k$. For technical reasons, we prefer to use a variant of the Bloch-Floquet transform introduced by Zak \cite{Zak:dynamics_Bloch_electrons:1968} which maps $\Psi \in \Schwartz(\R^d_x)$ onto u, 
\begin{align}
	( \Zak \Psi )(k,y) := \sum_{\gamma \in \Gamma} e^{- i k \cdot (y + \gamma)} \, \Psi(y + \gamma) 
	\label{rewrite:Zak:eqn:Zak_transform} 
	. 
\end{align}
The \BZak transform has the following periodicity properties: 
\begin{align}
	( \Zak \Psi ) (k - \gamma^*,y) &= e^{+ i \gamma^* \cdot y} \, ( \Zak \Psi ) (k,y) 
	=: \tau(\gamma^*) \, ( \Zak \Psi ) (k,y) 
	&& \forall \gamma^* \in \Gamma^* 
	\\ 
	( \Zak \Psi ) (k,y - \gamma) &= ( \Zak \Psi ) (k,y) 
	&& \forall \gamma \in \Gamma
	\notag 
\end{align}
$\tau$ is a unitary representation of the group of dual lattice translations $\Gamma^*$. By density, $\Zak$ immediately extends to $L^2(\R^d_x)$ and it maps it unitarily onto 
\begin{align}
	\Htau := \Bigl \{ \psi \in L^2_{\mathrm{loc}} \bigl ({\R^d_k},L^2(\T^d_y) \bigr ) \; \big \vert \;  \psi (k-\gamma^\ast) = \tau(\gamma^\ast) \, \psi(k) \mbox{ a.~e. } \forall \gamma^* \in \Gamma^* \Bigr \} 
	\label{rewriting:Zak:eqn:defnition_Htau} 
	, 
\end{align}
which is equipped with the scalar product 
\begin{align*}
	\bscpro{\varphi}{\psi}_{\tau} := \int_{\BZ} \dd k \, \bscpro{\varphi(k)}{\psi(k)}_{L^2(\T^d_y)} 
	. 
\end{align*}
It is obvious from the definition that the left-hand side does not depend on the choice of the unit cell $\BZ$ in reciprocal space. The \BZak representation of momentum $-i \nabla_x$ and position operator $\hat{x}$ on $L^2(\R^d_x)$, equipped with the obvious domains, can be computed directly, 
\begin{align}
	\Zak (- i \nabla_x) \Zak^{-1} &= \id_{L^2(\BZ)} \otimes (- i \nabla_y) + \hat{k} \otimes \id_{L^2(\T^d_y)} 
	\equiv -i \nabla_y + k 
	\label{rewrite:Zak:eqn:Zak_transformed_building_blocks}
	\\
	\Zak \hat{x} \Zak^{-1} &= i \nabla_k^{\tau} 
	, 
	\notag 
\end{align}
where we have used the identification $\Htau \cong L^2(\BZ) \otimes L^2(\T^d_y)$. The superscript $\tau$ on $i \nabla_k^{\tau}$ indicates that the operator's domain $\Htau \cap H^1_{\mathrm{loc}} \bigl ( \R^d , L^2(\T^d_y) \bigr )$ consists of $\tau$-equivariant functions. The \BZak transformed domain for momentum $-i \nabla_y + k$ is $L^2(\BZ) \otimes H^1(\T^d_y)$. Since the phase factor $\tau$ depends on $y$, the \BZak transform of $\hat{x}$ does not factor --- unless we consider $\Gamma$-periodic functions, then we have 
\begin{align*}
	\Zak V_{\Gamma}(\hat{x}) \Zak^{-1} = \id_{L^2(\BZ)} \otimes V_{\Gamma}(\hat{y}) \equiv V_{\Gamma}(\hat{y}) 
	. 
\end{align*}
Equations~\eqref{rewrite:Zak:eqn:Zak_transformed_building_blocks} immediately give us the \BZak transform of $\hat{H}$, namely 
\begin{align}
	\hat{H}^{\Zak} := \Zak \hat{H} \Zak^{-1} = \tfrac{1}{2} \bigl ( -i \nabla_y + k - \lambda A(i \eps \nabla_k^{\tau}) \bigr )^2 + V_{\Gamma}(\hat{y}) + \phi(i \eps \nabla_k^{\tau}) 
	\label{rewrite:Zak:eqn:Zak_transformed_hamiltonian} 
	, 
\end{align}
which defines an essentially selfadjoint operator on $\Zak \Cont^{\infty}_0(\R^d_x)$. If the external electromagnetic field vanishes, the hamiltonian 
\begin{align}
	\hat{H}_{\mathrm{per}}^{\Zak} := \Zak \hat{H}_{\mathrm{per}} \Zak^{-1} = \int_{\BZ}^{\oplus} \dd k \, \Hper^{\Zak}(k) 
\end{align}
fibers into a family of operators on $L^2(\T^d_y)$ indexed by crystal momentum $k \in \BZ$. $\tau$-equivariance relates $\Hper^{\Zak}(k - \gamma^*)$ and $\Hper^{\Zak}(k)$ via 
\begin{align*}
	\Hper^{\Zak}(k - \gamma^*) = \tau(\gamma^*) \, \Hper^{\Zak}(k) \, \tau(\gamma^*)^{-1} 
	&& \forall \gamma^* \in \Gamma^* 
\end{align*}
which, among other things, ensures that Bloch bands $\{ E_n \}_{n \in \N}$, \ie the solutions to the eigenvalue equation 
\begin{align*}
	\Hper^{\Zak}(k) \varphi_n(k) = E_n(k) \, \varphi_n(k) 
	, 
	&& \varphi_n(k) \in L^2(\T^d_y)
	, 
\end{align*}
are $\Gamma^*$-periodic functions. Standard arguments show that $\Hper^{\Zak}(k)$ has purely discrete spectrum for all $k \in \BZ$ and if Bloch bands are ordered by magnitude, they are smooth functions away from band crossings. Similarly, the Bloch functions $k \mapsto \varphi_n(k)$ are smooth if the associated energy band $E_n$ does not intersect with or touch others \cite{Reed_Simon:M_cap_Phi_4:1978}. 

The next subsection shows that the effect of introducing an external electromagnetic field can be interpreted as ``replacing'' the direct integral with the \emph{magnetic} quantization of $\Hper^{\Zak} + \phi$. 


\subsection{Magnetic $\Psi$DO and Weyl calculus} 
\label{rewriting:magnetic_weyl_calculus}
Instead of using regular Weyl calculus, we use a more sophisticated Weyl calculus that is adapted to magnetic problems. It has first been proposed by Müller in 1999 \cite{Mueller:productRuleGaugeInvariantWeylSymbols:1999} in a non-rigorous fashion. Independently, Măntoiu and Purice \cite{Mantoiu_Purice:magnetic_Weyl_calculus:2004} as well as Iftimie, Măntoiu and Purice \cite{Iftimie_Mantiou_Purice:magneticPseudodifferentialOperators:2005} have laid the mathematical foundation. All of the main results of ordinary Weyl calculus have been transcribed to the magnetic context; for details, we refer to the two aforementioned publications and those we give in the remainder of this section.

\subsubsection{Ordinary magnetic Weyl calculus} 
\label{rewriting:magnetic_weyl_calculus:mag_wc_on_Rd}
The basic building blocks of magnetic pseudodifferential operators are 
\begin{align}
	\PA \equiv \mathsf{P}^A_{\eps,\lambda} &:= - i \nabla_x - \lambda A(\Qe) 
	\\
	\Qe \equiv \Qe_{\eps} &:= \eps \hat{x} 
	\notag 
	. 
\end{align}
With this notation, $\opHamiltonian$ can be written in terms of $\PA$ and $\Qe$ as the quantization of 
\begin{align}
	\Hamiltonian(x,\xi) := \tfrac{1}{2} \xi^2 + V_{\Gamma} (\nicefrac{x}{\eps}) + \phi(x) 
	, 
\end{align}
\ie $\opHamiltonian = \Hamiltonian(\Qe,\PA)$. As $\Hamiltonian$ is the sum of a contribution \emph{quadratic} in momentum and a contribution depending only on $x$, this prescription is unambiguous. However, not all objects (\eg resolvents and projections) we will encounter are of this type. We need a functional calculus for the non-commuting family of operators $\Qe$ and $\PA$ of noncommutative observables that are characterized by the commutation relations 
\begin{align}
	i \bigl [ \Qe_l , \Qe_j \bigr ] = 0 
	&&
	i \bigl [ \Qe_l , \PA_j \bigr ] = \eps \, \delta_{lj} 
	&&
	i \bigl [ \PA_l , \PA_j \bigr ] = \eps \lambda B_{lj}(\Qe) 
	\label{rewriting:magnetic_weyl_calculus:eqn:comm_rel_Rd} 
	. 
\end{align}
The commutation relations can be rigorously implemented via the Weyl system 
\begin{align*}
	\WeylSys^A(x,\xi) := e^{- i (\xi \cdot \Qe - x \cdot \PA)} =: e^{- i \sigma((x,\xi),(\Qe,\PA))}
\end{align*}
where $A$ is a smooth, polynomially bounded vector potential associated to a magnetic field with components in $\BCont^{\infty}(\R^d)$ and $\sigma \bigl ( (x,\xi) , (y,\eta) \bigr ) := \xi \cdot y - x \cdot \eta$ is the non-magnetic symplectic form. We will also introduce the symplectic Fourier transform 
\begin{align*}
	(\Fs f)(x,\xi) := \frac{1}{(2\pi)^d} \int_{\R^d_y} \dd y \int_{\R^d_{\eta}} \dd \eta \, e^{ i \sigma((x,\xi),(y,\eta))} \, f(y,\eta) 
	, 
	&& 
	f \in \Schwartz(T^* \R^d_x) 
	, 
\end{align*}
which is also its own inverse on $\Schwartz(T^* \R^d_x)$ and extends to a continuous bijection on the space of tempered distributions $\Schwartz'(T^* \R^d_x)$. The Weyl quantization of $h \in \Schwartz(T^* \R^d_x)$ given by 
\begin{align}
	\Opx^A(h) := \frac{1}{(2\pi)^d} \int_{\R^d_x} \dd x \int_{\R^d_{\xi}} \dd \xi \, (\Fs h)(x,\xi) \, \WeylSys^A(x,\xi) 
	\label{rewriting:magnetic_weyl_calculus:eqn:OpxA}
\end{align}
defines a \emph{magnetic} $\Psi$DO. Associated to this, we have a product $\magW$ akin to the usual Moyal product which emulates the product of magnetic operators on the level of functions on phase space, $\Opx^A \bigl ( f \magW g \bigr ) = \Opx^A(f) \, \Opx^A(g)$. For suitable functions $f , g : T^* \R^d_x \longrightarrow \C$, \eg Hörmander-class symbols, their magnetic product is given by the oscillatory integral
\begin{align}
	(f \magW g)(x,\xi) &= \frac{1}{(2\pi)^{2d}} \int_{\R^d_y} \dd y \int_{\R^d_{\eta}} \dd \eta \int_{\R^d_z} \dd z \int_{\R^d_{\zeta}} \dd \zeta \, e^{i \sigma((x,\xi),(y,\eta) + (z,\zeta))} \, e^{i \frac{\eps}{2} \sigma((y,\eta),(z,\zeta))} 
	\cdot \notag \\
	&\qquad \qquad \qquad \qquad \cdot
	e^{- i \eps \gamma^B_{\eps}(x,y,z)} \, (\Fs f)(y,\eta) \, (\Fs g)(z,\zeta) 
	\label{rewriting:magnetic_weyl_calculus:eqn:magnetic_product}
\end{align}
where $\gamma^B_{\eps}(x,y,z)$ is the scaled magnetic flux through a triangle whose corners depend on $x$, $y$ and $z$. In \cite{Lein:two_parameter_asymptotics:2008} it was shown that for Hörmander-class symbols, this product has an asymptotic development in $\eps$ and $\lambda$, 
\begin{align}
	f \magW g = f \, g - \eps \tfrac{i}{2} \{ f , g \}_{\lambda B} + \order(\eps^2) 
	, 
	\label{rewriting:magnetic_weyl_calculus:eqn:magnetic_product_expansion}
\end{align}
where 
\begin{align*}
	\{ f , g \}_{\lambda B} := \sum_{l = 1}^d \bigl ( \partial_{\xi_l} f \, \partial_{x_l} g - \partial_{x_l} f \, \partial_{\xi_l} g \bigr ) - \lambda \sum_{l , j = 1}^d B_{lj} \, \partial_{\xi_l} f \, \partial_{\xi_j} g 
\end{align*}
is the magnetic Poisson bracket. In the limit $\lambda \rightarrow 0$ it reduces to the standard Poisson bracket. The crucial fact that this product depends on the magnetic field $B$ rather than the vector potential $A$ can be traced back to the \emph{gauge-covariance} of magnetic Weyl quantization: if $A' = A + \dd \chi$ is an equivalent vector potential, $\dd A = B = \dd A'$, then the quantizations with respect to either vector potential are unitarily equivalent, 
\begin{align}
  \Opx^{A + \dd \chi}(h) = e^{+ i \lambda \chi(\Qe)} \, \Opx^A(h) \, e^{- i \lambda \chi(\Qe)}
  .
\end{align}
\emph{This is generically false if we quantize $h_A(x,\xi) := h \bigl ( x , \xi - \lambda A(x) \bigr )$ via usual, non-magnetic Weyl quantization, $\Opx_A(h) := \Opx(h_A)$.}\footnote{Coincidentally, the quantization of polynomials of degree $\leq 2$ in momentum with respect to $\Opx_A$ \emph{are} covariant. Bloch bands, however, are not quadratic functions -- as are all the other objects (terms of the development of the projection and the unitary) involved in this paper. } Otherwise, the properties of the relevant pseudodifferential operators were to depend on the choice of gauge. Fortunately, the difference starts to appear at second order in $\eps$ \cite[Section~1.1.2]{Lein:two_parameter_asymptotics:2008}, so we generically expect results derived by nonmagnetic Weyl calculus to agree to first order. However, there is a second advantage of magnetic Weyl calculus (beyond being more natural): we can treat more general magnetic fields as properties of $B$ enter rather than those of $A$ --- and associated vector potentials are always worse behaved than the magnetic field. With some effort, one can treat magnetic fields which admit a vector potentials whose derivatives are all bounded, $\babs{\partial^a_x A(x)} \leq C_a$ for all $a \in \N_0^d$, $\abs{a} \geq 1$,  see~\cite[Section~5]{Iftimie_Mantiou_Purice:magneticPseudodifferentialOperators:2005}, for instance, but with magnetic Weyl calculus we are instantly able to treat magnetic fields whose components are $\BCont^{\infty}(\R^d)$ functions with \emph{zero extra effort}. The limitation to this class of fields is due to the fact that we are interested in Hörmander class symbols which need to be bounded in the $x$ variable. 
In fact, the extension of the results of Panati, Spohn and Teufel to magnetic fields of class $\BCont^{\infty}$ is our main motivation as explained in the introduction. Covariance ensures that our results do not depend on the choice of a nice or symmetric gauge, any vector potential that is a smooth, polynomially bounded function will do.

Many standard results of pseudodifferential theory have been transcribed to the magnetic case for a large class of magnetic fields: typically, it is either assumed that the components of $B$ are $\Cont^{\infty}_{\mathrm{pol}}(\R^d)$ or $\BCont^{\infty}(\R^d)$ functions, although we shall always assume the latter. The quantization and dequantization have been extended to tempered distributions and it was shown that smooth, uniformly polynomially bounded functions on phase space $\Cont^{\infty}_{\mathrm{pol} \, u}(T^* \R^d)$ are among those with good composition properties \cite[Proposition~23]{Mantoiu_Purice:magnetic_Weyl_calculus:2004}.  Hörmander symbols are preserved under the magnetic Weyl product and quantizations of real-valued, elliptic Hörmander symbols of positive order $m$ define selfadjoint operators on the $m$th magnetic Sobolev space \cite{Iftimie_Mantiou_Purice:magneticPseudodifferentialOperators:2005}. A magnetic version of the Caldéron-Vaillancourt theorem \cite{Iftimie_Mantiou_Purice:magneticPseudodifferentialOperators:2005} and commutator criteria \cite{MantoiuPurice:BealsCriterion:2008} show the interplay between properties of magnetic pseudodifferential operators and their associated symbols. 

Lastly, we mention something that will be important in the next section: the magnetic Weyl quantization is a position representation for a magnetic pseudodifferential operator. However, equivalently, we can use the momentum representation where $x$ is quantized to $\Qe':= i \eps \nabla_{\xi}$ and $\xi$ to ${\mathsf{P}'}^A := \hat{\xi} - \lambda A(i \eps \nabla_{\xi})$. The commutation relations of the building block operators in momentum representation are again encoded into the Weyl system 
\begin{align*}
	\WeylSys^{\prime \, A}(x,\xi) := e^{- i \sigma((x,\xi),(\Qe',{\mathsf{P}'}^A))} = \mathfrak{F} \WeylSys^A(x,\xi) \mathfrak{F}^{-1}
\end{align*}
which is related to the Weyl system in momentum representation via the Fourier transform $\mathfrak{F} : L^2(\R^d_x) \longrightarrow L^2(\R^d_{\xi})$. If $\Opx^{\prime \, A}$ is the quantization associated to the Weyl system $\WeylSys^{\prime \, A}$, then $\Opx^{\prime \, A}$ and $\Opx^A$ are related via $\mathfrak{F}$ as well. An important consequence is that the Weyl product is independent of the choice of representation: for suitable distributions $f$ and $g$, we conclude that the Weyl products must agree:
\begin{align*}
	\Opx^{\prime \, A} (f \Weyl^{\prime \, B}_{\eps,\lambda} g) &= \Opx^{\prime \, A}(f) \, \Opx^{\prime \, A}(g)
	\\
	&= \mathfrak{F} \, \Opx^A(f) \, \mathfrak{F}^{-1} \mathfrak{F} \, \Opx^A(g) \, \mathfrak{F}^{-1} = \mathfrak{F} \, \Opx^A(f \magWel g) \, \mathfrak{F}^{-1} 
	\\
	&= \Opx^{\prime \, A}(f \magWel g) 
\end{align*}
This is related to an algebraic point of view proposed in \cite{MantoiuPuriceRichard:twistedXProducts:2004} and elaborated upon in \cite{Lein_Mantoiu_Richard:anisotropic_mag_pseudo:2009}: the choice of Hilbert space in this case can be seen as a choice of (equivalent) representation of more fundamental $C^*$-algebras of distributions. Properties such as boundedness of the quantization of certain classes of distributions and the form of the composition law are preserved if we choose a unitarily equivalent representation. In fact, we could have replaced $\mathfrak{F}$ in the above argument by any other unitary operator $\mathcal{U} : L^2(\R^d_x) \longrightarrow \mathcal{H}$ where the target space $\mathcal{H}$ is again a separable Hilbert space. Gauge transformations are another particular example of unitarities which connect equivalent representations. 


\subsubsection{Equivariant magnetic Weyl calculus} 
\label{rewriting:magnetic_weyl_calculus:mag_wc_on_Td}
For technical reasons, we must adapt magnetic Weyl calculus to deal with equivariant, unbounded operator-valued functions. We follow the general strategy outlined in \cite{PST:effDynamics:2003}, but we need to be more careful as the roles of $\Qe$ and $\PA$ are not interchangeable if $B \neq 0$. We would like to reuse results for Weyl calculus on $T^* \R^d$ -- in particular, the two-parameter expansion of the product (equation~\eqref{rewriting:magnetic_weyl_calculus:eqn:magnetic_product_expansion}). Consider the building block kinetic operators \emph{macroscopic position} $\Reps$ and \emph{magnetic crystal momentum} $\KA$, 
\begin{align}
	\Reps &= i \eps \nabla_k \otimes \id_{L^2(\T^d_y)} \equiv i \eps \nabla_k 
	\label{rewriting:magnetic_weyl_calculus:eqn:building_blocks_Td} 
	\\
	\KA &= \hat{k} - \lambda A(\Reps) 
	\notag 
	, 
\end{align}
in momentum representation: they define selfadjoint operators whose domains are dense in $L^2_{\tau'} \bigl ( \R^d_k , L^2(\T^d_y) \bigr )$ where $\tau'$ stands for either $\tau : \gamma^* \mapsto e^{- i \gamma^* \cdot \hat{y}}$ or $1 : \gamma^* \mapsto 1$. The elements of this Hilbert space can be considered as vector-valued tempered distributions with special properties as $L^2_{\tau'} \bigl ( \R^d_k , L^2(\T^d_y) \bigr )$ can be continuously embedded into $\Schwartz' \bigl ( \R^d_k , L^2(\T^d_y) \bigr )$. For simplicity, let us ignore questions of domains and assume that $h \in \BCont^{\infty} \bigl ( T^* \R^d_x , \mathcal{B} \bigl ( L^2(\T^d_y) \bigr ) \bigr )$ is a \emph{bounded operator}-valued function. Then its magnetic Weyl quantization 
\begin{align}
	\Opk^A (h) := \frac{1}{(2\pi)^d} \int_{\R^d_r} \dd r \int_{\R^d_k} \dd k \, (\Fs h)(k,r) \, \WeylSys^A(k,r) 
	\label{rewrite:magnetic_weyl_calculus:eqn:mag_wq_torus}
\end{align}
defines a continuous operator from $\Schwartz \bigl ( \R^d_k , L^2(\T^d_y) \bigr )$ to itself which has a continuous extension as an operator from $\Schwartz' \bigl ( \R^d_k , L^2(\T^d_y) \bigr )$ to itself \cite[Proposition~21]{Mantoiu_Purice:magnetic_Weyl_calculus:2004}. Here, the corresponding Weyl system 
\begin{align*}
	\WeylSys^A(k,r) := e^{- i \sigma((r,k),(\Reps,\KA))} \otimes \id_{L^2(\T^d_y)} \equiv e^{- i (k \cdot \Reps - r \cdot \KA)}
\end{align*}
is defined in terms of the building block operators $\KA$ and $\Reps$ and acts trivially on $L^2(\T^d_y)$. The Weyl product $f \magW g$ of two suitable distributions associated to the quantization $\Opk^A$ is also given by a suitable reinterpretation of equation~\eqref{rewriting:magnetic_weyl_calculus:eqn:magnetic_product} as $f$ and $g$ are now operator-valued functions. Furthermore, we can also develop $f \magW g$ asymptotically in $\eps$ and $\lambda$ \cite[Theorem~1.1]{Lein:two_parameter_asymptotics:2008}. 
To see this, we remark that the difference between the products associated to $\Opx^A$ and $\Opk^A$ is two-fold: first of all, $\Opx^A$ is a position representation while $\Opk^A$ is a momentum representation. Let $\Opk^{\prime \, A}$ be the magnetic Weyl quantization defined with respect to $\Reps' := \mathfrak{F}^{-1} \Reps \mathfrak{F} = \eps \hat{r}$ and $\mathsf{K}^{\prime \, A} := \mathfrak{F}^{-1} \KA \mathfrak{F} = - i \nabla_r - \lambda A(\eps \hat{r})$, \ie the position representation. As explained at the end of the previous subsection, unitarily equivalent representations, here $\Opk^A$ and $\Opk^{\prime \, A}$, have the same Weyl product. 

Secondly, the functions which are to be quantized by $\Opx^A$ and $\Opk^A$ take values in $\C$ and the bounded operators on $L^2(\T^d_y)$, respectively. The interested reader may check the proofs regarding the various properties of the product $\magW$ in \cite{Mantoiu_Purice:magnetic_Weyl_calculus:2004,Iftimie_Mantiou_Purice:magneticPseudodifferentialOperators:2005} and \cite{Lein:two_parameter_asymptotics:2008} can be generalized to accommodate operator-valued functions, including Hörmander symbols. 
\begin{definition}[Hörmander symbols $\Hoerr{m}{\rho} \bigl ( \mathcal{B}(\Hil_1,\Hil_2) \bigr )$]
	Let $m \in \R$, $\rho \in [0,1]$ and $\Hil_1$, $\Hil_2$ be separable Hilbert spaces. Then a function $f$ is said to be in $\Hoerr{m}{\rho} \bigl ( \mathcal{B}(\Hil_1,\Hil_2) \bigr )$ if and only if for all $a , \alpha \in \N_0^d$ the seminorms 
	\begin{align*}
		\bnorm{f}_{m , a \alpha} := \sup_{(x,\xi) \in \PSpace} \Bigl ( \sqrt{1 + \xi^2} \Bigr )^{\sabs{\alpha} \rho-m} \bnorm{\partial_x^a \partial_{\xi}^{\alpha} f(x,\xi)}_{\mathcal{B}(\Hil_1,\Hil_2)} < \infty 
	\end{align*}
	are finite where $\snorm{\cdot}_{\mathcal{B}(\Hil_1,\Hil_2)}$ denotes the operator norm on $\mathcal{B}(\Hil_1,\Hil_2)$. In case $\rho = 1$, one also writes $\Hoer{m} := \Hoermr{m}{1}$
\end{definition}
Hörmander symbols which have an expansion in $\eps$ that is uniform in the small parameter are called semiclassical. 
\begin{definition}[Semiclassical symbols $\SemiHoermr{m}{\rho} \bigl ( \mathcal{B}(\Hil_1,\Hil_2) \bigr )$]\label{rewriting:magnetic_weyl_calculus:defn:semiclassical_symbol}
	A map $f : [0,\eps_0) \longrightarrow \Hoerr{m}{\rho}$, $\eps \mapsto f_{\eps}$ is called a semiclassical symbol of order $m \in \R$ and weight $\rho \in [0,1]$, that is $f \in \SemiHoermr{m}{\rho}$, if there exists a sequence $\{ f_n \}_{n \in \N_0}$, $f_n \in \Hoerr{m - n \rho}{\rho}$, such that for all $N \in \N_0$, one has 
	\begin{align*}
		\eps^{-N} \left ( f_{\eps} - \sum_{n = 0}^{N-1} \eps^n \, f_n \right ) \in \Hoerr{m - N \rho}{\rho} 
	\end{align*}
	uniformly in $\eps$ in the sense that for any $N \in \N_0$ and $a , \alpha \in \N_0^d$, there exist constants $C_{N a \alpha} > 0$ such that 
	\begin{align*}
		\norm{f_{\eps} - \sum_{n = 0}^{N-1} \eps^n \, f_n}_{m - N \rho , a \alpha} \leq C_{N a \alpha} \, \eps^N 
	\end{align*}
	holds for all $\eps \in [0,\eps_0)$. If $\rho = 1$, then one abbreviates $\SemiHoermr{m}{1}$ with $\SemiHoer{m}$. 
\end{definition}
Lastly, we will need the notion of $\tau$-equivariant symbols. 
\begin{definition}[$\tau$-equivariant symbols $\SemiTau \bigl ( \mathcal{B}(\Hil_1 , \Hil_2) \bigr )$]\label{rewriting:magnetic_weyl_calculus:defn:tau_equivariant_symbol}
	Let $\tau_j : \Gamma^* \longrightarrow \mathcal{U}(\Hil_j)$, $j = 1,2$, be unitary $*$-representations of the group $\Gamma^*$. Then $f \in \SemiHoermr{0}{0}$ is $\tau$-equivariant, \ie an element of $\SemiTau \bigl ( \mathcal{B}(\Hil_1,\Hil_2) \bigr )$, if and only if 
	\begin{align*}
		f(k - \gamma^* , r) = \tau_2(\gamma^*) \, f(k,r) \, \tau_1(\gamma^*)^{-1} 
	\end{align*}
	holds for all $k \in \R^d_k$, $r \in \R^d_r$ and $\gamma^* \in \Gamma^*$. 
	
\end{definition}
\medskip

\noindent
Now the reader is in a position to translate the results derived in Appendix~B of \cite{Teufel:adiabaticPerturbationTheory:2003} to the context of magnetic Weyl calculus, the modifications are straightforward and all the necessary references have been given in this section. 




\section{The magnetic Bloch electron as a space-adiabatic problem} 
\label{mag_sapt}
Our tool of choice to derive effective dynamics is space-adiabatic perturbation theory \cite{PST:sapt:2002,PST:effDynamics:2003,Teufel:adiabaticPerturbationTheory:2003} which uses pseudodifferential techniques to derive perturbation expansions order-by-order in a systematic fashion. We adapt their results by replacing ordinary Weyl calculus with \emph{magnetic} Weyl calculus. Adiabatic decoupling only hinges on $\eps \ll 1$ and does not rely on $\lambda$ to be small.

\subsection{The three ingredients of space-adiabatic problems} 
\label{mag_sapt:slow_var}
The insight of \cite{PST:effDynamics:2003} was that the slow variation of the external electromagnetic field (quantified by $\eps \ll 1$) leads to a decoupling into slow (macroscopic) and fast (microscopic) degrees of freedom. This is characteristic of adiabatic systems whose three main features are 
\begin{enumerate}[(i)]
	\item A distinction between \emph{slow} and \emph{fast degrees of freedom}: the original (physical) Hilbert space $\Hil = L^2(\R^d_x)$ is decomposed unitarily into $\Hslow \otimes \Hfast := L^2(\BZ) \otimes L^2(\T^d_y)$ in which the unperturbed hamiltonian is block diagonal (see diagram~\eqref{mag_sapt:diagram:unperturbed}). 
	
	The names slow and fast Hilbert space are due to the operators defined on them: on the fast Hilbert space, the two conjugate observables are $-i \nabla_y$ and $\hat{y}$ acting on $\Hfast = L^2(\T^d_y)$, the Hilbert space associated to the Wigner Seitz cell $\WS \cong \T^d_y$ (cf.~equation~\eqref{intro:eqn:WS}); their commutator is of $\order(1)$. The operators $\Reps$ and $\KA$ (cf.~equation~\eqref{rewriting:magnetic_weyl_calculus:eqn:building_blocks_Td}) defined on $\Hslow = L^2(\BZ)$ are considered slow, because their commutator is of $\order(\eps)$. Since $\Hslow$ is the Hilbert space over the Brillouin zone $\BZ$ (cf{.} equation~\eqref{intro:eqn:BZ}), the dynamics of the slow variables $(\Reps,\KA)$ describes the motion across unit cells in momentum representation whereas the dynamics of the fast variables $(\hat{y},-i \nabla_y)$ describe what happens within the Wigner Seitz cell $\WS$. 
	\item A \emph{small, dimensionless parameter $\eps$} that quantifies the separation of spatial scales. In our situation, $\eps \ll 1$ relates the variation of the external electromagnetic field to the microscopic scale as given by the lattice constant. In addition, we have the parameter $\lambda$. However, only the semiclassical parameter $\eps$ is crucial for adiabatic decoupling. 
	\item A \emph{relevant part of the spectrum}, \ie a subset of the spectrum which is separated from the remainder by a gap. We are interested in the dynamics associated to a family of Bloch bands $\{ E_n \}_{n \in \mathcal{I}}$ that does not intersect or merge with bands from the remainder of the spectrum.
\end{enumerate}
\begin{assumption}[Gap condition]\label{mag_sapt:mag_wc:defn:gap_condition}
The spectrum of $\hat{H}_{\mathrm{per}}^{\Zak}$ satisfies the \emph{gap condition}, namely there exists a family of Bloch bands $\{ E_n \}_{n \in \mathcal{I}}$, $\mathcal{I} = [I_- , I_+] \cap \N_0$ such that 
 \begin{align*}
     \inf_{k \in \BZ} \mathrm{dist} \Bigl ( \bigcup_{n \in \mathcal{I}} \{ E_n(k) \} , \bigcup_{j \not\in \mathcal{I}} \{ E_j(k) \} \Bigr ) =: C_g > 0
     .
 \end{align*}
\end{assumption}
The spectral gap ensures that transitions from and to the relevant part of the spectrum are exponentially suppressed. Band crossings within the relevant part of the spectrum are admissible, though.
\medskip

\noindent
In the original publication, an additional assumption was made on the existence of a \emph{smooth}, $\tau$-equivariant basis, a condition that is equivalent to the triviality of a certain $U(N)$ bundle over the torus $\T^d_k$ where $N := \abs{\mathcal{I}}$ is the number of bands including multiplicity. At least for the physically relevant cases, \ie $d \leq 3$, Panati has shown that this is always possible for nonmagnetic Bloch bands \cite{Panati:trivialityBlochBundle:2006}. For $d \geq 4$, our results still hold if we add 
\begin{assumption}[Smooth frame ($d \geq 4$)]\label{mag_sapt:mag_wc:assumption:smooth_frame}
	If $d \geq 4$, we assume there exists an orthonormal basis (called \emph{smooth frame}) $\{ \varphi_j(\cdot) \}_{j=1, \ldots , \abs{\mathcal{I}}}$ of whose elements are smooth and $\tau$-equivariant with respect to $k$, \ie $\varphi_j(\cdot-\gamma^{\ast}) = \tau(\gamma^{\ast}) \varphi_j(\cdot)$ for all $\gamma^{\ast} \in \Gamma^{\ast}$ and for all  $j \in \{ 1 , \ldots , \abs{\mathcal{I}} \}$. 
\end{assumption}
%

\subsection{Rewriting the unperturbed problem: an adiabatic point of view} 
\label{mag_sapt:the_unperturbed_problem_from_an_adiabatic_point_of_view}
Let us consider the unperturbed case, \ie in the absence of an external electromagnetic field. Then the dynamics on $\Htau$ is generated by $\hat{H}_{\mathrm{per}}^{\Zak} = \int_{\BZ}^{\oplus} \dd k \, \Hper^{\Zak}(k)$. Each fiber hamiltonian $\Hper^{\Zak}(k)$ is an operator on the fast Hilbert space $\Hfast = L^2(\T^d_y)$. Then $\hat{\pi}_0 = \int_{\BZ}^{\oplus} \dd k \, \pi_0(k)$ is the projection onto the relevant part of the spectrum, where 
\begin{align*}
 \pi_0(k) := \sum_{n \in \mathcal{I}} \sopro{\varphi_n(k)}{\varphi_n(k)}
 .
\end{align*}
Even though the $\varphi_n(k)$ may not be continuous at eigenvalue crossings, the projection $k \mapsto \pi_0(k)$ is due to the spectral gap. Associated to the relevant band is a (non-unique) unitary $\hat{u}_0 = \int_{\BZ}^{\oplus} \dd k \, u_0(k)$ which ``straightens'' $\Htau$ into $L^2(\BZ_k) \otimes L^2(\T^d_y)$: for each $k \in \BZ$, we define
\begin{align*}
 u_0(k) := \sum_{n \in \mathcal{I}} \sopro{\chi_n}{\varphi_n(k)} + u_0^{\perp}(k)
\end{align*}
where $\chi_n \in L^2(\T^d_y)$, $n \in \Index$, are \emph{fixed} vectors \emph{independent of $k$} and $u_0^{\perp}(k)$ (also non-unique) acts on the complement of $\mathrm{ran} \, \pi_0(k)$ and is such that $\hat{u}_0$ is a proper unitary. Even though this means $u_0$ is not unique, the specific choices of the $\{ \chi_n \}_{n \in \Index}$ and $u_0^{\perp}$ will not enter the derivation. Then we can put all parts of the puzzle into a diagram:
\begin{align}
	\bfig
		\node L2R(-900,0)[L^2(\R^d_x)]
		\node piL2R(-900,-600)[\Zak^{-1} \hat{\pi}_0 \Zak L^2(\R^d_x)]
		\node Htau(0,0)[\Htau]
		\node piHtau(0,-600)[\hat{\pi}_0 \Htau]
		\node HslowHfast(900,0)[L^2(\BZ) \otimes L^2(\T^d_y)]
		\node Href(900,-600)[L^2(\BZ) \otimes \C^N]
		\arrow[L2R`Htau;\Zak]
		\arrow/-->/[L2R`piL2R;\Zak^{-1} \hat{\pi}_0 \Zak]
		\arrow/-->/[piL2R`piHtau;]
		\arrow[Htau`piHtau;\hat{\pi}_0]
		\arrow[Htau`HslowHfast;\hat{u}_0]
		\arrow[HslowHfast`Href;\Piref]
		\arrow/-->/[piHtau`Href;]
		\Loop(-900,0){L^2(\R^d_x)}(ur,ul)_{e^{-i \frac{t}{\eps} \hat{H}}} 
		\Loop(0,0){\Htau}(ur,ul)_{e^{-i \frac{t}{\eps} \hat{H}^{\Zak}}} 
		\Loop(900,-600){L^2(\BZ) \otimes \C^N}(dr,dl)^{e^{-i \frac{t}{\eps} \hat{h}_{\mathrm{eff} \, 0}}} 
	\efig
	\label{mag_sapt:diagram:unperturbed}
\end{align}
The reference projection $\Piref = \id_{L^2(\BZ)} \otimes \piref$ acts trivially on the first factor, $L^2(\BZ)$, and projects via 
\begin{align}
	\piref = \sum_{j = 1}^N \sopro{\chi_j}{\chi_j} = u_0(k) \, \pi_0(k) \, u_0^*(k)
	\label{mag_sapt:eqn:piref}
\end{align}
onto an $N$-dimensional subspace of $L^2(\T^d_y)$. We will identify $\piref L^2(\T^d_y)$ with $\C^N$ when convenient and in this sense, we identify the range of $\Piref$ with the reference space $\Href := L^2(\BZ) \otimes \C^N$. 

The dynamics in the lower-right corner is generated by the effective hamiltonian
\begin{align*}
 \hat{h}_{\mathrm{eff} \, 0} := \Piref \, \hat{u}_0 \, \hat{H}_{\mathrm{per}}^{\Zak} \, \hat{u}_0^* \, \Piref
\end{align*}
which reduces to $E_n(\hat{k})$ if the relevant part of the spectrum consists of an isolated Bloch band. 

\subsection{Adiabatic decoupling in the presence of external fields} 
\label{mag_sapt:adiabatic_decoupling}
Now the question is whether a similar diagram exists even if the perturbation is present, \ie if there exist a \emph{tilted projection} $\Pi$, an \emph{intertwining unitary} $U$ and an \emph{effective hamiltonian} $\hat{h}_{\mathrm{eff}}$ that take the place of $\hat{\pi}_0$, $\hat{u}_0$ and $\hat{h}_{\mathrm{eff} \, 0}$? This has been answered in the positive for magnetic fields that admit $\BCont^{\infty}(\R^d,\R^d)$ vector potentials in \cite{PST:effDynamics:2003} where these objects are explicitly constructed by recursion. We replace standard Weyl calculus used in the original publication with its magnetic variant (cf{.} Section~\ref{rewriting:magnetic_weyl_calculus}) which naturally allows for the treatment of more general magnetic fields with components in $\BCont^{\infty}$. 

The construction of $\Pi$ and $U$ detailed in the next section is a ``defect construction'' where recursion relations derived from
\begin{align*}
	{\Pi}^2 &= \Pi
	&&
	\bigl [ \Pi , \hat{H}^{\Zak} \bigr ] = 0 \\
	{U}^* \, U &= \id_{\Htau}, \; U \, {U}^* = \id_{L^2_{\mathrm{per}}(\BZ) \otimes L^2_{\mathrm{per}}(\WS)}
	&&
	U \, \Pi \, {U}^* = \Piref 
\end{align*}
relate the $n$th term to all previous terms. These four conditions merely characterize that $\Pi$ and $U$ are still a projection and a unitary (first column) and adapted to the problem (second column). These equations can be translated via magnetic Weyl calculus to
\begin{align}
	\pi \magW \pi &= \pi + \ordere{\infty}
	&&
	\bigl [ \pi , H^{\Zak} \bigr ]_{\magW} = \ordere{\infty}
	\label{mag_sapt:eqn:condition_corrected_objects} \\
	u \magW {u}^* &= 1 + \ordere{\infty} = {u}^* \magW u
	&&
	u \magW \pi \magW {u}^* = \pi_{\mathrm{eff}} + \ordere{\infty}
	\notag 
\end{align}
where 
\begin{align}
	\Hamiltonian^{\Zak}(k,r) := \tfrac{1}{2} \bigl ( -i \nabla_y + k \bigr )^2 + V_{\Gamma}(\hat{y}) + \phi(r) 
	\label{mag_sapt:eqn:symbol_H_Zak} 
\end{align}
is the operator-valued symbol associated to the magnetic pseudodifferential operator $\opHamiltonian^{\Zak}$ defined by equation~\eqref{rewrite:Zak:eqn:Zak_transformed_hamiltonian}. Note that the magnetic vector \emph{potential} $A$ does not enter the definition of the symbol $\Hamiltonian^{\Zak}$ and we do not need to impose conditions on $A$ to ensure the symbol $\Hamiltonian^{\Zak}$ is well-behaved after minimal substitution. This is why magnetic Weyl calculus can be used to treat much more general magnetic fields. 

For technical reasons, $\Opk^A(u)$ and $U$, for instance, agree only up to an error that is arbitrarily small in $\eps$ with respect to the operator norm,
\begin{align*}
	U = \Opk^A(u) + \ordern(\eps^{\infty})
	.
\end{align*}
The tilted projection and intertwining unitary are now used to define the effective hamiltonian as the magnetic quantization of 
\begin{align*}
	\heff := \piref \, u \magW \Hamiltonian^{\Zak} \magW u^* \, \piref 
\end{align*}
which generates effective dynamics, \ie for initial states in $\Pi \Htau$ we can approximate the full time evolution in terms of $e^{- i \frac{t}{\eps} \Opk^A(\heff)}$. In turn, the effective quantum evolution can be approximated by semiclassical dynamics. Theorem~\ref{eom:thm:semiclassical_limit}, the main result of the next section, will make this statement precise. 



\section{Derivation of effective quantum dynamics} 
\label{effective_quantum_dynamics}
The aforementioned ``defect construction'' yields the tilted projection $\pi$ and the intertwining unitary $u$ as asymptotic expansion in $\eps$ and $\lambda$. It is important that the decoupling is \emph{solely due to the separation of spatial scales} quantified by $\eps$ and \emph{independent} of $\lambda$ which regulates the strength of the magnetic field.

\subsection{The dynamics in the almost invariant subspace} 
\label{effective_quantum_dynamics:effective_dynamics}
We will quickly explain how $\Pi$ and $U$ are computed order-by-order in $\eps$ and $\lambda$. We adapt the general recipe explained in \cite{Teufel:adiabaticPerturbationTheory:2003} to incorporate two parameters: since the decoupling is due to the separation of spatial scales quantified by $\eps \ll 1$, we will order corrections in powers of $\eps$ first. Expanding the magnetic Weyl product to zeroth order, we can check 
\begin{align*}
	\pi_0 \magW \pi_0 &= \pi_0 + \ordere{}
	&&
	\bigl [ \pi_0 , H^{\Zak} \bigr ]_{\magW} = \ordere{} 
	\\
	u_0 \magW {u_0}^* &= 1 + \ordere{} = {u_0}^* \magW u_0
	&&
	u_0 \magW \pi_0 \magW u_0^* = \pi_{\mathrm{eff}} + \ordere{} 
	. 
\end{align*}
Here, $\bigl [ \pi_0 , H^{\Zak} \bigr ]_{\magW} := \pi_0 \magW H^{\Zak} - H^{\Zak} \magW \pi_0$ denotes the magnetic Weyl commutator. The asymptotic expansion of the product is key to deriving corrections in a systematic manner: the $\ordere{}$ terms can be used to infer $\pi_1$ and $u_1$, the subprincipal symbols. Then, one proceeds by recursion: if $\pi^{(n)} := \sum_{l = 0}^n \eps^l \pi_l$ and $u^{(n)} := \sum_{l = 0}^n \eps^l u_l$ satisfy equations~\eqref{mag_sapt:eqn:condition_corrected_objects} up to errors of order $\eps^{n+1}$, then we can compute $\pi_{n+1}$ and $u_{n+1}$. The construction of $\pi$ and $u$ follows exactly from Lemma 3.8 and Lemma 3.15 of \cite{Teufel:adiabaticPerturbationTheory:2003}; it is \emph{purely algeraic} and only uses that we have a recipe to expand the Moyal product in terms of the semiclassical parameter $\eps$. Let us define 
\begin{align}
	\pi^{(n)} \magW \pi^{(n)} - \pi^{(n)} &=: \eps^{n+1} G_{n+1} + \ordere{n+2} 
	\label{effective_quantum_dynamics:effective_dynamics:eqn:projection_defects} \\
	\bigl [ \Hamiltonian^{\Zak} , \pi^{(n)} + \eps^{n+1} \pi_{n+1}^{\mathrm{d}} \bigr ]_{\magW} &=: \eps^{n+1} F_{n+1} + \ordere{n+2} 
	\notag 
\end{align}
as projection and commutation defects and 
\begin{align}
	u^{(n)} \magW {u^{(n)}}^* - 1 &=: \eps^{n+1} A_{n+1} + \ordere{n+2} 
	\label{effective_quantum_dynamics:effective_dynamics:eqn:unitarity_defects} \\
	\bigl ( u^{(n)} + \eps^{n+1} a_{n+1} u_0 \bigr ) \magW \pi^{(n+1)} \magW \bigl ( u^{(n)} + \eps^{n+1} a_{n+1} u_0 \bigr )^* &=: \eps^{n+1} B_{n+1} + \ordere{n+2} 
	\notag
\end{align}
as unitarity and intertwining defects. The diagonal part of the projection $\pi_{n+1}^{\mathrm{d}}$ can be computed from $G_{n+1}$ via 
\begin{align}
	\pi_{n+1}^{\mathrm{D}} := - \pi_0 G_{n+1} \pi_0 + (1 - \pi_0) G_{n+1} (1 - \pi_0) 
	\label{effective_quantum_dynamics:effective_dynamics:eqn:pi_n1_diag} 
	. 
\end{align}
The term 
\begin{align}
	a_{n+1} = - \tfrac{1}{2} A_{n+1} 
	\label{effective_quantum_dynamics:effective_dynamics:eqn:u_a} 
\end{align}
stems from the ansatz $u_{n+1} = (a_{n+1} + b_{n+1}) u_0$ where $a_{n+1}$ and $b_{n+1}$ are symmetric and antisymmetric, respectively. One can solve the second equation for 
\begin{align}
	b_{n+1} = \bigl [ \piref , B_{n+1} \bigr ] 
	\label{effective_quantum_dynamics:effective_dynamics:eqn:u_b} 
\end{align}
where $\piref$ is the reference projection on $L^2(\T^d_y)$ given by equation~\eqref{mag_sapt:eqn:piref}. This equation fixes only the off-diagonal part of $b_{n+1}$ as $\piref B_{n+1} \piref = 0 = (1 - \piref) B_{n+1} (1 - \piref)$ and in principle one is free to choose the diagonal part of $b_{n+1}$. This means, there is a freedom that allows arbitrary unitary transformations within $\piref L^2(\T^d_y)$ as well as its orthogonal complement. In general, it is not possible to solve 
\begin{align}
	\bigl [ \Hamiltonian^{\Zak} , \pi_{n+1}^{\mathrm{OD}} \bigr ] = - F_{n+1} 
\end{align}
explicitly since Bloch functions at band crossings \emph{within} the relevant part of the spectrum (which are admissible) are no longer differentiable. In any case, $\pi$ can be constructed locally around $(k_0,r_0)$ by asymptotically expanding the Moyal resolvent $(\Hamiltonian^{\Zak} - z)^{(-1)_B}$, \ie the symbol defined through the relations 
\begin{align*}
	(\Hamiltonian^{\Zak} - z) \magW (\Hamiltonian^{\Zak} - z)^{(-1)_B} = 1 = (\Hamiltonian^{\Zak} - z)^{(-1)_B} \magW (\Hamiltonian^{\Zak} - z) 
	, 
\end{align*}
and setting 
\begin{align}
	\pi(k,r) = \frac{i}{2 \pi} \int_{C(k_0,r_0)} \dd z \, (\Hamiltonian^{\Zak} - z)^{(-1)_B}(k,r) + \ordere{\infty} 
\end{align}
in a neighborhood of $(k_0,r_0)$. A recent result by Iftime, Măntoiu and Purice \cite{MantoiuPurice:BealsCriterion:2008} suggests that under these circumstances ($\Hamiltonian^{\Zak}$ is elliptic and selfadjoint operator-valued) $(\Hamiltonian^{\Zak} - z)^{(-1)_B}$ always exists and is a Hörmander symbol even in the presence of a magnetic field. We reckon their result extends to the case of operator-valued symbols, but seeing how tedious the proof is, we simply stick to the procedure used by Panati, Spohn and Teufel \cite[Lemma~5.17]{Teufel:adiabaticPerturbationTheory:2003}. This construction uniquely fixes the tilted Moyal projection $\pi$, but not the Moyal unitary $u$. 

As the \emph{two-parameter expansion of the product} 
\begin{align*}
	f \magW g \asymp \sum_{n = 0}^{\infty} \sum_{k = 0}^n \eps^n \lambda^k \, (f \magW g)_{(n,k)} 
	, 
\end{align*}
contributes only \emph{finitely many terms} in $\lambda$ for fixed power of $n$ of $\eps$ \cite{Lein:two_parameter_asymptotics:2008}, we can order the terms of the expansion of $\pi$ and $u$ in powers of $\lambda$, \eg 
\begin{align*}
	\pi_n = \sum_{k = 0}^n \lambda^k \, \pi_{(n,k)} 
	. 
\end{align*}
The magnetic Weyl product as well as its asymptotic expansion are defined in terms of \emph{oscillatory integrals}, \ie integrals which exist in the distributional sense. If we take the limit $\lambda \rightarrow 0$ of $f \magW g$, we can interchange oscillatory integration and limit procedure \cite[p.~90]{Hoermander:fourierIntOp1:1972} and conclude $\lim_{\lambda \rightarrow 0} f \magW g = f \Weyl g$ where $\Weyl$ is the usual Moyal product. Similarly, we can apply this reasoning to the asymptotic expansion: for any fixed $N \in \N_0$, we may write the product as 
\begin{align*}
	f \magW g = \sum_{n = 0}^N \eps^n \biggl ( \sum_{k = 0}^n \lambda^k \, (f \magW g)_{(n,k)} \biggr ) + \eps^{N+1} \, R_{N+1}^B(f,g) 
\end{align*}
and taking the limit $\lambda \rightarrow 0$ means only the nonmagnetic terms $(f \magW g)_{(n,0)}$ survive. The remainder also behaves nicely when taking the limit as it is also just another oscillatory integral and $\lim_{\lambda \rightarrow 0} R_{N+1}^B(f,g)$ is exactly the remainder of the nonmagnetic Weyl product. 
\medskip

\noindent
Hence, we can now prove the main result of this paper: 
\begin{theorem}[Effective quantum dynamics]\label{effective_quantum_dynamics:effective_dynamics:thm:adiabatic_decoupling}
	Let Assumptions~\ref{intro:assumption:V_Gamma}, \ref{intro:assumption:em_fields} and~\ref{mag_sapt:mag_wc:defn:gap_condition} be satisfied. Furthermore, if $d \geq 4$, we add Assumption~\ref{mag_sapt:mag_wc:assumption:smooth_frame}. Then there exist 
	\begin{enumerate}[(i)]
		\item an orthogonal projection $\Pi \in \mathcal{B}(\Htau)$, 
		\item a unitary map $U$ which intertwines $\Htau$ and $L^2(\BZ) \otimes L^2(\T^d_y)$, and 
		\item a selfadjoint operator $\Opk^A(\heff) \in \mathcal{B} \bigl ( L^2(\BZ) \otimes \C^N \bigr)$, $N := \abs{\Index}$ 
	\end{enumerate}
	such that 
	\begin{align}
		\bnorm{\bigl [ \opHamiltonian^{\Zak} , \Pi \bigr ]} = \ordere{\infty} 
		\label{effective_quantum_dynamics:effective_dynamics:eqn:commutator_Pi_HZak}
	\end{align}
	and
	\begin{align}
		\bnorm{\bigl ( e^{- i s \opHamiltonian^{\Zak}} - U^* e^{- i s \Opk^A(\heff)} U \bigr ) \Pi}_{\mathcal{B}(\Htau)} = \mathcal{O} \bigl ( \eps^{\infty}(1 + \abs{s}) \bigr ) 
		\label{effective_quantum_dynamics:effective_dynamics:eqn:adiabatic_decoupling_dynamics}
		. 
	\end{align}
	The effective hamiltonian is the magnetic quantization of the $\Gamma^*$-periodic symbol 
	\begin{align}
		\heff := \piref \, u \magW \Hamiltonian^{\Zak} \magW u^* \, \piref \asymp \sum_{n = 0}^{\infty} \eps^n \, \heff_n \in \SemiHoer{0}_{\tau \equiv 1} \bigl ( \mathcal{B}(\C^N) \bigr ) 
	\end{align}
	whose asymptotic expansion can be computed to any order in $\eps$ and $\lambda$. To each order in $\eps$, only finitely many terms in $\lambda$ contribute, $\heff_n = \sum_{k = 0}^n \lambda^k \, \heff_{(n,k)}$.
\end{theorem}
One deduces from equation~\eqref{effective_quantum_dynamics:effective_dynamics:eqn:commutator_Pi_HZak} and a Duhamel argument that the unitary time evolution generated by $\opHamiltonian^{\Zak}$ and $\Pi$ almost commute even for macroscopic times $t = \eps s$ and hence, up to an error of arbitrarily large order in $\eps$, the space $\Piref \Htau$ is left invariant by the dynamics, 

\begin{align*}
	(1 - \Pi) e^{- i \frac{t}{\eps} \opHamiltonian^{\Zak}} \Pi = \order_{\norm{\cdot}}(\eps^{\infty} \abs{t}) 
	. 
\end{align*}
Thus we call $\Pi \Htau$ the \emph{almost invariant subspace} and in a sense, it is the tilted ``eigenspace'' associated to the family of relevant bands. 

The proof of the above theorem amounts to showing (i)-(iii) separately. 
\begin{proposition}[Tilted projection]\label{effective_quantum_dynamics:effective_dynamics:prop:projection}
	Under the assumptions of Theorem~\ref{effective_quantum_dynamics:effective_dynamics:thm:adiabatic_decoupling} there exists and orthogonal projection $\Pi \in \mathcal{B}(\Htau)$ such that 
	\begin{align}
		\bigl [ \opHamiltonian^{\Zak} , \Pi \bigr ] = \ordern(\eps^{\infty}) 
	\end{align}
	and $\Pi = \Opk^A(\pi) + \ordern(\eps^{\infty})$ where $\Opk^A(\pi)$ is the magnetic Weyl quantization of a $\tau$-equivariant semiclassical symbol 
	\begin{align*}
		\pi \asymp \sum_{n = 0}^{\infty} \eps^n \, \pi_n \in \SemiTau \bigl ( \mathcal{B}(\Hfast) \bigr )
	\end{align*}
	whose principal part $\pi_0(k,r)$ coincides with the spectral projection of $\Hamiltonian^{\Zak}(k,r)$ onto the subspace corresponding to the given isolated family of Bloch bands $\{ E_n \}_{n \in \Index}$. Each term in the expansion can be written as a finite sum 
	\begin{align*}
		\pi_n = \sum_{k = 0}^n \lambda^k \, \pi_{(n,k)} \in \SemiTau \bigl ( \mathcal{B}(\Hfast) \bigr ) 
	\end{align*}
	ordered by powers of $\lambda$. For $\lambda \rightarrow 0$, the projection $\pi$ reduces to the nonmagnetic projection $\pi^0 \asymp \sum_{n = 0}^{\infty} \eps^n \, \pi_{(n,0)}$. 
\end{proposition}
\begin{proof}[Sketch]
	The proof relies on a well-developed \emph{magnetic} Weyl calculus adapted to operator-valued symbols (cf.~Section~\ref{rewriting:magnetic_weyl_calculus}) and the gap condition. In particular, one needs a magnetic Caldéron-Vaillancourt theorem, composition and quantization of Hörmander symbols \cite{Iftimie_Mantiou_Purice:magneticPseudodifferentialOperators:2005} and finally, an asymptotic two-parameter expansion of the magnetic Weyl product $\magW$ \cite{Lein:two_parameter_asymptotics:2008}. The interested reader may check line-by-line that the original proof \cite[Proposition~5.16]{Teufel:adiabaticPerturbationTheory:2003} can be transliterated to the magnetic context with obvious modifications. If we were using standard Weyl calculus, the major obstacle would be to control derivatives of $\pi$ since vector potentials may be unbounded. In magnetic Weyl calculus the vector potential at no point enters the calculuations and the assumptions on the magnetic field assure that $\pi \in \SemiTau \bigl ( \mathcal{B}(\Hfast) \bigr )$ is a proper $\tau$-equivariant semiclassical Hörmander-class symbol (cf.~Definition~\ref{rewriting:magnetic_weyl_calculus:defn:tau_equivariant_symbol}). 
	
	The fact that we can write all of the $\pi_n$ as \emph{finite} sum of terms ordered by powers of $\lambda$ stems from the fact that calculating $\pi_n$ involves the expansion of the product up to $n$th power in $\eps$, \eg for the projection defect, we find 
	\begin{align*}
		\pi^{(n-1)} \magW &\pi^{(n-1)} - \pi^{(n-1)} = \eps^n \negmedspace \negmedspace \negmedspace \negmedspace \sum_{a + b + c = n} \negmedspace \negmedspace (\pi_a \magW \pi_b \bigr )_{(c)} + \ordere{n+1}
		\\
		&= \eps^n \negmedspace \negmedspace \negmedspace \negmedspace \sum_{a + b + c = n} \sum_{a' = 0}^a \sum_{b' = 0}^b \sum_{c' = 0}^c \lambda^{a' + b' + c'} \, \bigl ( \pi_{(a,a')} \magW \pi_{(b,b')} \bigr )_{(c,c')} + \ordere{n+1} 
		. 
	\end{align*}
	Certainly, the exponent of $\lambda$ is always bounded by $n \geq a' + b' + c'$. And since the sum is finite, this clearly defines a semiclassical symbol in $\eps$ is shown (cf.~Definition~\ref{rewriting:magnetic_weyl_calculus:defn:semiclassical_symbol}). Similar arguments for the commutation defect in conjunction with the comments in the beginning of this section show $\pi$ to be a semiclassical symbol. 
	It is well-behaved under the $\lambda \rightarrow 0$ limit and reduces to the projection associated to the case $B = 0$. 
	\medskip
	
	\noindent
	Lastly, to make the almost projection $\Opk^A(\pi)$ into a true projection, we define $\Pi$ to be the spectral projection onto the spectrum in the vicinity of $1$, 
	\begin{align*}
		\Pi := \int_{\abs{z - \nicefrac{1}{2}} = 1} \dd z \, \bigl ( \Opk^A(\pi) - z \bigr )^{-1} 
		. 
	\end{align*}
	This concludes the proof. 
\end{proof}
Similarly, one can modify the proof of Proposition~5.18 found in \cite{Teufel:adiabaticPerturbationTheory:2003} to show the existence of the intertwining unitary. 
\begin{proposition}[Intertwining unitary]\label{effective_quantum_dynamics:effective_dynamics:prop:unitary}
	Let $\{ E_n \}_{n \in \Index}$ be a family of bands separated by a gap from the others and let Assumption~\ref{intro:assumption:V_Gamma} be satisfied. If $d > 3$, assume $u_0 \in \Hoermr{0}{0} \bigl ( \mathcal{B}(\Hfast) \bigr )$. Then there exists a unitary operator $U : \Htau \longrightarrow L^2(\BZ) \otimes L^2(\T^d_y)$ such that $U = \Opk^A(u) + \ordern(\eps^{\infty})$ where 
	\begin{align*}
		u \asymp \sum_{n = 0}^{\infty} \eps^n u_n \in \SemiHoer{0} \bigl ( \mathcal{B}(\Hfast) \bigr )
	\end{align*}
	is right-$\tau$-covariant at any order and has principal symbol $u_0$. Each term in the expansion can be written as a finite sum 
	\begin{align*}
		u_n = \sum_{k = 0}^n \lambda^k u_{(n,k)} 
	\end{align*}
	ordered by powers of $\lambda$. For $\lambda \rightarrow 0$, the unitary $u$ reduces to the nonmagnetic unitary $u^0 \asymp \sum_{n = 0}^{\infty} \eps^n \, u_{(n,0)}$. 
\end{proposition}
\begin{proof}[Sketch]
	Equations~\eqref{effective_quantum_dynamics:effective_dynamics:eqn:u_a} and \eqref{effective_quantum_dynamics:effective_dynamics:eqn:u_b} give us $a_{n+1}$ and $b_{n+1}$ which combine to $u_{n+1} = (a_{n+1} + b_{n+1}) u_0$; by Theorem~1.1 from \cite{Lein:two_parameter_asymptotics:2008} it is also in the correct symbol class, namely $\Hoermr{0}{0} \bigl ( \mathcal{B}(\Hfast) \bigr )$. The right $\tau$-covariance is also obvious from the ansatz. 
	
	Lastly, the true unitary $U$ is obtained via the Nagy formula as described in \cite{Teufel:adiabaticPerturbationTheory:2003}. 
\end{proof}
\begin{proof}[Theorem~\ref{effective_quantum_dynamics:effective_dynamics:thm:adiabatic_decoupling} (Sketch)]
	The existence of $\Pi$ and $U$ have been the subject of Propositions~\ref{effective_quantum_dynamics:effective_dynamics:prop:projection} and \ref{effective_quantum_dynamics:effective_dynamics:prop:unitary}. By right-$\tau$-covariance of $u$, $\heff$ is a $\Gamma^*$-periodic symbol and since it is the magnetic Weyl product of $\BCont^{\infty} \bigl ( T^* \R^d , \mathcal{B} \bigl ( L^2(\T^d_y) \bigr ) \bigr )$ functions, Theorem~1.1 from \cite{Lein:two_parameter_asymptotics:2008} ensures that the product and its asymptotic two-parameter expansion are in $\BCont^{\infty} \bigl ( T^* \R^d , \mathcal{B} \bigl ( L^2(\T^d_y) \bigr ) \bigr )$ as well. Equation~\eqref{effective_quantum_dynamics:effective_dynamics:eqn:adiabatic_decoupling_dynamics} follows as usual from a Duhamel argument. 
\end{proof}
The crucial statement of Theorem~\ref{effective_quantum_dynamics:effective_dynamics:thm:adiabatic_decoupling} is equation~\eqref{effective_quantum_dynamics:effective_dynamics:eqn:adiabatic_decoupling_dynamics} and it is worthwhile to discuss its implications to applications: in practice, one only computes finitely many terms of the asymptotic expansions of $\Pi$, $U$ and $\heff$. Let us call their finite resummations $u^{(l)} := \sum_{n = 0}^l \eps^n \, u_n$, $\pi^{(l)} := \sum_{n = 0}^l \eps^n \, \pi_n$ and $\heff^{(l)} := \sum_{n = 0}^l \eps^n \, \heff_n$. Assume we are interested in times $t = \eps^k s$ where $\abs{t} \leq \tau$. Then a closer inspection of the Duhamel argument in the proof of Theorem~\ref{effective_quantum_dynamics:effective_dynamics:thm:adiabatic_decoupling} yields 
\begin{align*}
	e^{- i \frac{t}{\eps^k} \Hamiltonian^{\Zak}} \Pi = \Opk^A \bigl ( {u^{(n)}}^* \bigr ) \, e^{- i \frac{t}{\eps^k} \Opk^A(\heff^{(n+k)})} \, \Piref \, \Opk^A \bigl ( u^{(n)} \bigr ) + \order_{\norm{\cdot}} \bigl ( \eps^{n-k+1} \bigr )
	. 
\end{align*}
Hence, if we would like to consider macroscopic times $t = \eps s$ and make an error in the propagation of order $\order(\eps^2)$, then we need to expand the tilted projection, intertwining unitary and the effective hamiltonian to second order. However, if the relevant band consists of a single band, computing the first-order correction to the effective hamiltonian suffices as we shall see in Section~\ref{eom}. 


\subsection{Effective dynamics for a single band: the Peierls substitution} 
\label{effective_quantum_dynamics:effective_dynamics_for_a_single_band}
In case the relevant part of the spectrum consists of a single non-degenerate band $\Eb$, we can calculate the first-order correction to $\heff$ explicitly: the magnetic Weyl product reduces to the pointwise product to zeroth order in $\eps$. Thus, we can directly compute 
\begin{align*}
	\heff_0 = \piref \, u_0 \, H_0 \, u_0^* \, \piref =: \piref \, h_0 \, \piref = \Eb + \phi 
	. 
\end{align*}
For the first order, we use the recursion formula \cite[eq.~(3.35)]{Teufel:adiabaticPerturbationTheory:2003} and the fact that $\heff_0$ is a scalar-valued symbol: 
\begin{align*}
	\heff_1 &= \Bigl ( u_1 \, H_0 - h_0 \, u_1 + (u_0 \magW H_0)_{(1)} - (h_0 \magW u_0)_{(1)} \Bigr ) u_0^* 
	\\
	&= \piref \, \bigl [ u_1 u_0^* , h_0 \bigr ] \, \piref + (u_0 \magW H_0)_{(1)} u_0^* - (h_0 \magW u_0)_{(1)}  u_0^*
	\\
	&= - \tfrac{i}{2} \piref \, \Bigl ( \bigl \{ u_0 , H_0 \bigr \}_{\lambda B} - \bigl \{ h_0 , u_0 \bigr \}_{\lambda B} \Bigr ) \, \piref 
\end{align*}
The term with the magnetic Poisson bracket can be easily computed: 
\begin{align*}
	\piref \, \Bigl ( \bigl \{ u_0 , &H_0 \bigr \}_{\lambda B} - \bigl \{ h_0 , u_0 \bigr \}_{\lambda B} \Bigr ) \, u_0^* \, \piref 
	= 
	\\
	&= \piref \, \Bigl ( \partial_{k_l} u_0 \, \partial_{r_l} H_0 - \lambda B_{lj} \, \partial_{k_l} u_0 \, u_0^* \, 
	\bigl ( \partial_{k_j} h_0 \, u_0 - \partial_{k_j} u_0 \, u_0^* \, h_0 - h_0 \, u_0 \, \partial_{k_j} u_0^* \bigr )
	+ \\
	&\qquad \qquad 
	+ \partial_{r_l} h_0 \, \partial_{k_l} u_0 + \lambda B_{lj} \, \partial_{k_l} h_0 \, \partial_{k_j} u_0 \Bigr ) \, u_0^* \, \piref 
	\\
	&= 2 i \bigl ( \partial_{r_l} \phi - \lambda B_{lj} \, \partial_{k_j} \Eb \bigr ) \, \BerryC_l 
	+ \lambda B_{lj} \, \piref \, \partial_{k_l} u_0 \, u_0^* \, \bigl ( \partial_{k_j} u_0 \, u_0^* \, h_0 + h_0 \, u_0 \, \partial_{k_j} u_0^* \bigr ) \, \piref 
	\\
	&= 2 i \bigl ( \partial_{r_l} \phi - \lambda B_{lj} \, \partial_{k_j} \Eb \bigr ) \, \BerryC_l + \lambda B_{lj} \, \bscpro{\partial_{k_l} \varphi_{\mathrm{b}}}{\bigl ( \Hper - \Eb \bigr ) \partial_{k_j} \varphi_{\mathrm{b}}} 
\end{align*}
The first term combines to a Lorentz force term, the second one -- which is purely imaginary -- yields the Rammal-Wilkinson term. The components of the magnetic field are $\BCont^{\infty}(\R^d)$ functions and hence, principal and subprincipal symbol are in $\BCont^{\infty}(T^* \R^d)$ as well. 

This means, we have proven the following corollary to Theorem~\ref{effective_quantum_dynamics:effective_dynamics:thm:adiabatic_decoupling}: 
\begin{corollary}\label{effective_quantum_dynamics:effective_dynamics:cor:adiabatic_decoupling_explicit}
	Under the assumptions of Theorem~\ref{effective_quantum_dynamics:effective_dynamics:thm:adiabatic_decoupling}, the principal and subprincipal symbol of the effective hamiltonian for a single non-degenerate Bloch band $\Eb$ are given by 
	\begin{align}
		\heff_0 &= \Eb + \phi
		\\
		\heff_1 &= - \bigl ( - \partial_{r_l} \phi + \lambda B_{lj} \, \partial_{k_j} \Eb \bigr ) \, \BerryC_l - \lambda B_{lj} \, \mathcal{M}_{lj} 
		\notag \\
		&=: - F_{\mathrm{Lor} \, l} \, \BerryC_l - \lambda B_{lj} \, \mathcal{M}_{lj} 
		\notag 
		. 
	\end{align}
	where 
	\begin{align*}
		\BerryC_l(k) := i \scpro{\varphi_{\mathrm{b}}(k)}{\nabla_k \varphi_{\mathrm{b}}(k)} 
	\end{align*}
	and 
	\begin{align*}
		\mathcal{M}_{lj}(k) = \Re \Bigl ( \tfrac{i}{2} \bscpro{\partial_{k_l} \varphi_{\mathrm{b}}}{\bigl ( \Hper(k) - \Eb(k) \bigr ) \, \partial_{k_j} \varphi_{\mathrm{b}}} \Bigr )
	\end{align*}
	are the Berry connection and the so-called Rammal-Wilkinson term, respectively. 
\end{corollary}
To leading order, this is the well-known Peierls substitution. In particular, for zero electric field, $\phi \equiv 0$, and constant magnetic field (and the equations written in symmetric gauge), the ansatz 
\begin{align*}
	\Eb(k) = \sum_{j = 1}^d \cos k_j 
\end{align*}
yields the celebrated Hofstadter model \cite{Hofstadter:Bloch_electron_magnetic_fields:1976}. 


\section{Derivation of semiclassical equations of motion} 
\label{eom}
In the preceding section, we have approximated the full quantum evolution by a simpler effective \emph{quantum} evolution on a smaller reference space $\Href = L^2(\BZ) \otimes \C^N$. Now, we will link the time evolution generated by $\heff$ -- and thus $e^{- i \frac{t}{\eps} \opHamiltonian^{\Zak}}$ -- to a \emph{semiclassical flow} which contains quantum corrections. 
Conceptually, this is a two-step process: if we reconsider the diagram of spaces, 
\begin{align}
	\bfig
		\node L2R(-900,0)[L^2(\R^d_x)]
		\node piL2R(-900,-600)[\Zak^{-1} \Pi \Zak L^2(\R^d_x)]
		\node Htau(0,0)[\Htau]
		\node piHtau(0,-600)[\Pi \Htau]
		\node HslowHfast(900,0)[L^2(\BZ) \otimes L^2(\T^d_y)]
		\node Href(900,-600)[L^2(\BZ) \otimes \C^N]
		\arrow[L2R`Htau;\Zak]
		\arrow/-->/[L2R`piL2R;\Zak^{-1} \Pi \Zak]
		\arrow/-->/[piL2R`piHtau;]
		\arrow[Htau`piHtau;\Pi]
		\arrow[Htau`HslowHfast;U]
		\arrow[HslowHfast`Href;\Piref]
		\arrow/-->/[piHtau`Href;]
		\Loop(-900,0){L^2(\R^d_x)}(ur,ul)_{e^{-i \frac{t}{\eps} \opHamiltonian}} 
		\Loop(0,0){\Htau}(ur,ul)_{e^{-i \frac{t}{\eps} \opHamiltonian^{\Zak}}} 
		\Loop(900,-600){L^2(\BZ) \otimes \C^N}(dr,dl)^{e^{-i \frac{t}{\eps} \Opk^A(\heff)}} 
	\efig
	\label{eom:diagram:perturbed}
\end{align}
we notice that our physical observables live on the upper-left space $L^2(\R^d_x)$ -- or equivalently on $\Htau$. The effective evolution generated by $\heff$ approximates the dynamics if the initial states are localized in the almost invariant subspace associated to the relevant bands. In this section, we always assume the relevant part of the spectrum consists of a \emph{single non-degenerate band} $\Eb$ and thus $L^2(\BZ) \otimes \C^1 \cong L^2(\BZ)$. 

In a first step, we need to connect the semiclassical dynamics in the left column of the diagram with those in the lower-right corner. The second, much simpler step is to establish an Egorov-type theorem on the reference space $\Href = L^2(\BZ) \otimes \C^N$.

\subsection{Relation between dynamics for macroscopic and effective observables} 
\label{eom:connecting_dynamics}
Since we are concerned with the semiclassical dynamics of a particle in an electromagnetic field, the magnetic field must enter in the classical equations of motion. There are two ways: either one uses minimal coupling, \ie one writes down the equations of motion for position $r$ and kinetic momentum $k = k' - \lambda A(r)$ with respect to the usual symplectic form for the variables $(k',r)$ and the hamiltonian $\Hamiltonian^{\Zak} \bigl ( k' - \lambda A(r) , r \bigr )$. Then the classical flow which enters the Egorov theorem is generated by 
\begin{align}
	\left (
	\begin{matrix}
		\lambda B(r) & - \id \\
		+ \id & 0 \\
	\end{matrix}
	\right )
	\left (
	\begin{matrix}
		\dot{r} \\
		\dot{k} \\
	\end{matrix}
	\right )
	= 
	\left (
	\begin{matrix}
		\nabla_r \\
		\nabla_k \\
	\end{matrix}
	\right )
	\Hamiltonian^{\Zak}(k,r)
	\label{eom:eqn:initial_eom}
\end{align}
where the appearance of $B$ in the matrix representation of the symplectic form is due to the fact that $k$ is \emph{kinetic} momentum. 
What constitutes a suitable observable? Physically, we are interested in measurements on the \emph{macroscopic} scale, \ie the observable should be independent of the \emph{microscopic} degrees of freedom. On the level of symbols, this means $f(k,r)$ has to commute pointwise with the hamiltonian $\Hamiltonian^{\Zak}(k,r)$ for all $k$ and $r$. Hence, such an observable is a constant of motion with respect to the \emph{fast} dynamics. In the simplest case, the observables are scalar-valued. This also ensures we are able to ``separate'' the contributions to the full dynamics band-by-band. Note that this \emph{by no means} implies $\Opk^A(f)$ commutes with $\Opk^A(\Hamiltonian^{\Zak})$, but rather that all of the non-commutativity is contained in the slow variables $(k,r)$. 
\begin{definition}[Macroscopic semiclassical observable]\label{eom:defn:semiclassical_observable}
	A macropscopic observable $f$ is a scalar-valued semiclassical symbol (cf.~Definition~\ref{rewriting:magnetic_weyl_calculus:defn:semiclassical_symbol}) $\SemiHoermr{0}{0}(\C)$ which is $\Gamma^*$-periodic in $k$, $f(k + \gamma^*,r) = f(k,r)$ for all $(k,r) \in T^* \R^d$, $\gamma^* \in \Gamma^*$. 
\end{definition}
Our assumption that our dynamics lives on the almost-invariant subspace $\Pi \Htau$ modifies the classical dynamics to first order in $\eps$ as well: instead of using $\KA$ and $\Reps$ as building block observables, the proper observables should be $\Pi \KA \Pi$ and $\Pi \Reps \Pi$. Equivalently, we can switch to the reference space representation and use the magnetic quantization of 
\begin{align}
	\keff &:= \piref \, u \magW k \magW u^* \, \piref = k + \eps \lambda B(r) \BerryC(k) + \ordere{2}
	\label{eom:eqn:effective_building_blocks} \\
	\reff &:= \piref \, u \magW r \magW u^* \, \piref = r + \eps \BerryC(k) + \ordere{2} 
	\notag 
\end{align}
The crucial proposition we will prove next says that for suitable observables $f$, the effect of going to the effective representation is, up to errors of order $\eps^2$ at least, equivalent to replacing the arguments $k$ and $r$ by $\keff$ and $\reff$, 
\begin{align}
	\Piref U \Opk^A(f) U^* \Piref &= \Opk^A \bigl ( \piref \, u \magW f \magW u^* \, \piref \bigr ) + \ordern(\eps^{\infty}) 
	\notag \\
	&
	=: \Opk^A(f_{\mathrm{eff}}) + \ordern(\eps^{\infty}) 
	. 
\end{align}
Then it follows that the effective observable $f_{\mathrm{eff}}$ coincides with the original observable $f$ after a change of variables up to errors of order $\ordere{2}$, 
\begin{align}
	f_{\mathrm{eff}} = \piref \, u \magW f \magW u^* \, \piref = f \circ T_{\mathrm{eff}} + \ordere{2}
	, 
	\label{eom:eqn:effective_sym_replacement}
\end{align}
where the map $T_{\mathrm{eff}} : (k,r) \mapsto (\keff,\reff)$ maps the observables $k$ and $r$ onto the effective observables $\keff$ and $\reff$ defined via equations~\eqref{eom:eqn:effective_building_blocks} . 
\begin{proposition}\label{eom:prop:effective_sym_building_block}
	Let $f$ be a macroscopic semiclassical observable. Then up to errors of order $\eps^2$ equation~\eqref{eom:eqn:effective_sym_replacement} holds and consequently, we have 
	\begin{align}
		\Piref \, U \, \Opk^A(f) \, U^* \, \Piref = \Opk^A \bigl ( f_{\mathrm{eff}} \bigr ) + \ordern(\eps^{\infty}) 
		= \Opk^A \bigl ( f \circ T_{\mathrm{eff}} \bigr ) + \ordern(\eps^2) 
		\label{eom:eqn:effective_op_replacement} 
		. 
	\end{align}
\end{proposition}
\begin{proof}
	The equivalence of the left-hand sides of equations~\eqref{eom:eqn:effective_op_replacement} and \eqref{eom:eqn:effective_sym_replacement} follows from $U = \Opk^A(u) + \ordern(\eps^{\infty})$ and the fact that we only need to consider the first two terms in the $\eps$ expansion. With the help of Theorem~1.1 from \cite{Lein:two_parameter_asymptotics:2008}, we conclude $f_{\mathrm{eff}} \in \SemiHoer{0}$ is also a semiclassical symbol of order $0$. The left-hand side of \eqref{eom:eqn:effective_sym_replacement} can be computed explicitly: to zeroth order, nothing changes as $f$ commutes pointwise with $u$ and $u^*$, 
	\begin{align*}
		f_{\mathrm{eff} \, 0} = \piref \, u_0 f u_0^* \, \piref = f_0 
		. 
	\end{align*}
	To first order, we have 
	\begin{align*}
		f_{\mathrm{eff} \, 1} &= \piref \, \Bigl ( u_0 f_1 + u_1 f_0 - f_{\mathrm{eff} \, 0} u_1 + (u_0 \magW f_0)_{(1)} - (f_{\mathrm{eff} \, 0} \magW u_0)_{(1)} \Bigr ) u_0^* \, \piref 
		\\
		&= \piref \, u_0 f_1 u_0^* \, \piref - \tfrac{i}{2} \Bigl ( \bigl \{ u_0 , f_0 \bigr \}_{\lambda B} - 
		\bigl \{ f_{\mathrm{eff} \, 0} , u_0 \bigr \}_{\lambda B} \Bigr ) 
		\\
		&= f_1 - i \bigl ( \partial_{r_j} f_0 + \lambda B_{lj}(r) \, \partial_{k_l} f_0 \bigr ) \, \piref \, \partial_{k_j} u_0 \, u_0^* \, \piref 
		\\
		&=  f_1 + \bigl ( \partial_{r_j} f_0 + \lambda B_{lj}(r) \, \partial_{k_l} f_0 \bigr ) \, \BerryC_j 
		. 
	\end{align*}
	On the other hand, if we Taylor expand $f \circ T_{\mathrm{eff}} = f \bigl ( \keff , \reff \bigr )$ to first order in $\eps$, we get 
	\begin{align*}
		f \bigl ( \keff , \reff \bigr ) &= f_0 \bigl ( k + \eps \lambda B(r) \BerryC(k) + \ordere{2} , r + \eps \BerryC(k) + \ordere{2} \bigr ) + 
		\\
		&\qquad \qquad 
		+ \eps f_1 \bigl ( k + \eps \lambda B(r) \BerryC(k) + \ordere{2} , r + \eps \BerryC(k) + \ordere{2} \bigr ) + \ordere{2} 
		\\
		&= f_0(k,r) 
		+ \eps \Bigl ( f_1(k,r) + \lambda \partial_{k_l} f_0(k,r) \, B_{lj}(r) \, \BerryC_j(k) 
		+ \Bigr . \\
		&\qquad \qquad \qquad \qquad \Bigl . 
		+ \partial_{r_j} f_0(k,r) \, \BerryC_j(k) \Bigr ) 
		+ \ordere{2} 
	\end{align*}
	which coincides with $f_{\mathrm{eff}}$ up to $\ordere{2}$. 
\end{proof}
Now if the equations of motion~\eqref{eom:eqn:initial_eom} are an approximation of the full quantum dynamics, what are the equations of motion with respect to the effective variables? The classical flow $\Phi^{\mathrm{eff}}_t$ generated by $\heff$ with respect to the magnetic symplectic form (equation~\eqref{eom:eqn:initial_eom}) can be rewritten in terms of effective variables, 
\begin{align}
	\Phi^{\mathrm{macro}}_t = T_{\mathrm{eff}} \circ \Phi^{\mathrm{eff}}_t \circ T_{\mathrm{eff}}^{-1} + \ordere{2} 
	. 
	\label{eom:eqn:macro_flow}
\end{align}
The right-hand side \emph{does not serve} as a definition for the flow of the macroscopic observables, but it is a consequence: $\Phi^{\mathrm{macro}}_t$ is the flow associated to a modified symplectic form and a modified hamiltonian. The modified symplectic form includes the Berry curvature associated to $\Eb$ acting as a pseudo-magnetic field on the position variables. 
\begin{proposition}\label{eom:prop:macro_flow}
	Let $\Phi_t^{\mathrm{macro}}$ be the flow on $T^* \R^d$ generated by 
	\begin{align}
		\left (
		\begin{matrix}
			\lambda B(\reff) & - \id \\
			+ \id & \eps \Omega(\keff) \\
		\end{matrix}
		\right )
		\left (
		\begin{matrix}
			\dot{r}_{\mathrm{eff}} \\
			\dot{k}_{\mathrm{eff}} \\
		\end{matrix}
		\right )
		= 
		\left (
		\begin{matrix}
			\nabla_{\reff} \\
			\nabla_{\keff} \\
		\end{matrix}
		\right )
		h_{\mathrm{sc}}(\keff,\reff)
		\label{eom:eqn:macro_eom}
	\end{align}
	where the \emph{semiclassical hamiltonian} is given by 
	\begin{align}
		h_{\mathrm{sc}} := \heff \circ T_{\mathrm{eff}}^{-1} 
		. 
	\end{align}
	Then equation~\eqref{eom:eqn:macro_flow} holds, \ie 
	$\Phi_t^{\mathrm{macro}}$ and $T_{\mathrm{eff}} \circ \Phi_t^{\mathrm{eff}} \circ T_{\mathrm{eff}}^{-1}$ agree up to errors of order $\eps^2$. 
\end{proposition}
\begin{proof}
	We express $k$ and $r$ in terms of $\keff$ and $\reff$ in \eqref{eom:eqn:initial_eom} since, for $\eps$ small enough, $T_{\mathrm{eff}} : (k,r) \mapsto (\keff,\reff)$ is a bijection. For instance, the semiclassical hamiltonian $\heff \circ T_{\mathrm{eff}}^{-1}$ simplifies to 
	\begin{align*}
		h_{\mathrm{sc}}(\keff,\reff) :=& \, \heff \bigl ( \keff - \eps \lambda B(\reff) \BerryC(\keff) , \reff - \eps \BerryC(\keff) \bigr ) + \ordere{2} 
		\notag \\ 
		=& \, \bigl ( \Eb(\keff) + \phi(\reff) \bigr ) - \eps \lambda B(\reff) \cdot \mathcal{M}(\keff) + \ordere{2} 
		. 
	\end{align*}
	The symplectic form can be easily expanded to 
	\begin{align*}
		&\left (
		\begin{matrix}
			\lambda B(\reff - \eps \BerryC(\keff)) & - \id \\
			+ \id & 0 \\
		\end{matrix}
		\right ) 
		= 
		\\
		&\qquad \qquad 
		= 
		\left (
		\begin{matrix}
			\lambda B(\reff) & - \id \\
			+ \id & 0 \\
		\end{matrix}
		\right ) - \eps \left (
		\begin{matrix}
			\lambda \partial_{{\reff}_l} B(\reff) \, \BerryC(\keff) & 0 \\
			0 & 0 \\
		\end{matrix}
		\right ) 
		+ \ordere{2} 
		. 
	\end{align*}
	The other two terms, the time derivatives and gradients of $\keff$ and $\reff$ have slightly more complicated expansions, but they can be worked out explicitly. Then if we put all of them together, we arrive at the modified symplectic form~\eqref{eom:eqn:macro_eom}. This proves the first claim. 
	
	Hence, the hamiltonian vector fields agree up to $\ordere{2}$ and Lemma~5.24 in \cite{Teufel:adiabaticPerturbationTheory:2003} implies that also the flows differ only by $\ordere{2}$. 
\end{proof}
\begin{remark}
	These equations of motion have first been proposed in the appendix of \cite{PST:effDynamics:2003} and we have derived them in a more systematic fashion. The effective coordinates $r_{\mathrm{eff}}$ and $k_{\mathrm{eff}}$ are associated to a non-standard Poisson structure on $T^* \R^d_k$: from equation~\eqref{eom:eqn:macro_eom}, one can read off that the Poisson bracket with respect to $r_{\mathrm{eff}}$ and $k_{\mathrm{eff}}$ is given by 
	\begin{align*}
		\bigl \{ f , g \bigr \}_{\lambda B , \eps \Omega} = \sum_{l = 1}^d \bigl ( \partial_{\xi_l} f \, \partial_{x_l} g - \partial_{x_l} f \, \partial_{\xi_l} g \bigr ) - \sum_{l , j = 1}^d \bigl ( \lambda \, B_{lj} \, \partial_{\xi_l} f \, \partial_{\xi_j} g - \eps \, \Omega_{lj} \, \partial_{x_l} f \, \partial_{x_j} g \bigr )
	\end{align*}
	and thus different components of position $r_{\mathrm{eff}}$ no longer commute, 
	\begin{align*}
		\bigl \{ r_{\mathrm{eff} \, l} , r_{\mathrm{eff} \, j} \bigr \}_{\lambda B , \eps \Omega} &= - \eps \, \Omega_{lj} 
		. 
	\end{align*}
	Hence, $\Omega$ acts as a pseudomagnetic field that is due to quantum effects. 
\end{remark}
%


\subsection{An Egorov-type theorem} 
\label{eom:egorov}
The semiclassical approximation hinges on an Egorov-type theorem which we first prove on the level of effective dynamics: 
\begin{theorem}
	Let $\heff$ be the effective hamiltonian as given by Theorem~\ref{effective_quantum_dynamics:effective_dynamics:thm:adiabatic_decoupling} associated to an isolated, non-degenerate Bloch band $\Eb$. Then for any $\Gamma^*$-periodic semiclassical observable $f \in \SemiHoer{0}$, $f = f_0 + \eps f_1$, the flow $\Phi^{\mathrm{eff}}_t$ generated by $\heff$ with respect to the magnetic symplectic form (equation~\eqref{eom:eqn:initial_eom}) approximates the quantum evolution uniformly for all $t \in [-T , +T]$, 
	\begin{align}
		\Bnorm{e^{+ i \frac{t}{\eps} \Opk^A(\heff)} \, \Opk^A(f) \, e^{- i \frac{t}{\eps} \Opk^A(\heff)} - \Opk^A \bigl ( f \circ \Phi^{\mathrm{eff}}_t \bigr )}_{\mathcal{B}(L^2(\BZ))} \leq C \eps^2 
		. 
	\end{align}
\end{theorem}
\begin{proof}
	Since $\heff \in \BCont^{\infty}(T^* \R^d)$, the flow inherits the smoothness and $f \circ \Phi^{\mathrm{eff}}_t , \frac{\dd}{\dd t} \bigl ( f \circ \Phi^{\mathrm{eff}}_t \bigr ) \in \SemiHoer{0}(\C)$ remain also $\Gamma^*$-periodic in the momentum variable. To compare the two time-evlutions, we use the usual Duhammel trick which yields 
	\begin{align}
		e^{+ i \frac{t}{\eps} \Opk^A(\heff)} \, &\Opk^A(f) \, e^{- i \frac{t}{\eps} \Opk^A(\heff)} - \Opk^A \bigl ( f \circ \Phi^{\mathrm{eff}}_t \bigr ) =
		\notag \\
		&= \int_0^t \dd s \, \frac{\dd}{\dd s} \Bigl ( e^{+ i \frac{s}{\eps} \Opk^A(\heff)} \, \Opk^A \bigl ( f \circ \Phi^{\mathrm{eff}}_{t-s} \bigr ) \, e^{- i \frac{s}{\eps} \Opk^A(\heff)} \Bigr )
		\notag \\
		&= \int_0^t \dd s \, e^{+ i \frac{s}{\eps} \Opk^A(\heff)} \cdot
		\Bigl ( \tfrac{i}{\eps} \bigl [ \Opk^A(\heff) , \Opk^A \bigl ( f \circ \Phi^{\mathrm{eff}}_{t-s} \bigr ) \bigr ] 
		+ \notag \\
		&\qquad \qquad \qquad \qquad \qquad \qquad \qquad 
		- \Opk^A \bigl ( \tfrac{\dd}{\dd s} \bigl ( f \circ \Phi^{\mathrm{eff}}_{t-s} \bigr ) \bigr ) \Bigr ) \, e^{- i \frac{s}{\eps} \Opk^A(\heff)} 
		\notag \\
		&= \int_0^t \dd s \, e^{+ i \frac{s}{\eps} \Opk^A(\heff)} \, \Opk^A \Bigl ( \tfrac{i}{\eps} \bigl [ \heff , f \circ \Phi^{\mathrm{eff}}_{t-s} \bigr ]_{\magW} 
		+ \notag \\
		&\qquad \qquad \qquad \qquad \qquad \qquad \qquad 
		- \bigl \{ \heff , f \circ \Phi^{\mathrm{eff}}_{t-s} \bigr \}_{\lambda B} \Bigr ) \, e^{- i \frac{s}{\eps} \Opk^A(\heff)} 
		. 
		\label{eom:eqn:Egorov_Duhammel}
	\end{align}
	The magnetic Moyal commutator -- to first order -- agrees with the magnetic Poisson bracket, 
	\begin{align*}
		\tfrac{i}{\eps} \bigl [ \heff , f \circ \Phi^{\mathrm{eff}}_{t-s} \bigr ]_{\magW} &= \bigl \{ \heff , f \circ \Phi^{\mathrm{eff}}_{t-s} \bigr \}_{\lambda B} + \ordere{2} 
		\\
		&= \sum_{l = 1}^d \bigl ( \partial_{k_l} \heff \, \partial_{r_l} f - \partial_{r_l} \heff \, \partial_{k_l} f \bigr ) - \sum_{l,j = 1}^d B_{lj} \, \partial_{k_l} \heff \, \partial_{k_j} f + \ordere{2} 
		. 
	\end{align*}
	Hence, the term to be quantized in equation~\eqref{eom:eqn:Egorov_Duhammel} vanishes up to first order in $\eps$, 
	\begin{align*}
		\mbox{r.h.s.~of \eqref{eom:eqn:Egorov_Duhammel}} &= \int_0^t \dd s \, e^{+ i \frac{s}{\eps} \Opk^A(\heff)} \, \Opk^A \bigl ( 0 + \ordere{2} \bigr ) \, e^{- i \frac{s}{\eps} \Opk^A(\heff)} = \ordern(\eps^2)
		. 
	\end{align*}
	This finishes the proof. 
\end{proof}
The main result combines Proposition~\ref{eom:prop:effective_sym_building_block} with the Egorov theorem we have just proven: 
\begin{theorem}[Semiclassical limit]\label{eom:thm:semiclassical_limit}
	Let Assumptions~\ref{intro:assumption:V_Gamma}, \ref{intro:assumption:em_fields}, \ref{mag_sapt:mag_wc:defn:gap_condition} be satisfied; if $d \geq 4$, assume in addition that Assumption~\ref{mag_sapt:mag_wc:assumption:smooth_frame} holds true. Furthermore, let us assume the relevant part of the spectrum consists of a single non-degenerate Bloch band $\Eb$. Then for all macroscopic semiclassical observables $f$ (Definition~\ref{eom:defn:semiclassical_observable}) the full quantum evolution can be approximated by the hamiltonian flow $\Phi^{\mathrm{macro}}_t$ as given in Proposition~\ref{eom:prop:macro_flow} if the initial state is localized in the corresponding almost invariant subspace $\Zak^{-1} \Pi \Zak L^2(\R^d_x)$, 
	\begin{align}
		\Bnorm{\Zak^{-1} \Pi \Zak \, \Bigl ( e^{+ i \frac{t}{\eps} \opHamiltonian} \, \Opk^A(f) \, e^{- i \frac{t}{\eps} \hat{H}} - \Opk^A \bigl ( f \circ \Phi^{\mathrm{macro}}_t \bigr ) \Bigr ) \, \Zak^{-1} \Pi \Zak}_{\mathcal{B}(L^2(\R^d_x))} \leq C_T \eps^2 
		. 
	\end{align}
\end{theorem}
\begin{proof}
	We now combine all of these results to approximate the dynamics: let $f$ be a macroscopic observable. Then if we start with a state in the range of $\Zak^{-1} \Pi \Zak$, the time-evolved observable can be written as 
	\begin{align*}
		\Zak^{-1} \, \Pi \, &\Zak e^{- i \frac{t}{\eps} \hat{H}^{\eps}} \, \Zak^{-1} \Opk^A(f) \Zak \, e^{+ i \frac{t}{\eps} \hat{H}^{\eps}} \Zak^{-1} \Pi \Zak = 
		\\
		&= \Zak^{-1} \, \Pi e^{-i \frac{t}{\eps} \hat{H}^{\Zak}} \Opk^A(f) e^{+ i \frac{t}{\eps} \hat{H}^{\Zak}} \Pi \, \Zak 
		\\ 
		&= \Zak^{-1} \, U^{-1} \Piref U U^{-1} e^{- i \frac{t}{\eps} \hat{h}} U \Opk^A(f) U^{-1} e^{+ i \frac{t}{\eps} \hat{h}} U U^{-1} \Piref U \, \Zak + \ordern(\eps^{\infty})
		\\
		&= \Zak^{-1} \, U^{-1} \Piref e^{- i \frac{t}{\eps} \hat{h}} U \Opk^A(f) U^{-1} e^{+ i \frac{t}{\eps} \hat{h}} \Piref U \, \Zak + \ordern(\eps^{\infty})
		\\
		&= \Zak^{-1} \, U^{-1} \Piref e^{- i \frac{t}{\eps} \hat{h}_{\mathrm{eff}}} \Piref U \Opk^A(f) U^{-1} \Piref e^{+ i \frac{t}{\eps} \hat{h}_{\mathrm{eff}}} \Piref U \, \Zak + \ordern(\eps^{\infty})
		. 
	\end{align*}
	After replacing $U$ with $\Opk^A(u)$ (which adds another $\ordern(\eps^{\infty})$ error) and $\Piref$ with $\Opk^A(\piref)$, the term in the middle combines to the quantization of $f_{\mathrm{eff}} = \piref \, u \magW f \magW u^* \, \piref$. We apply Proposition~\ref{eom:prop:effective_sym_building_block} and the Egorov theorem involving $\heff$ and obtain 
	\begin{align*}
		\ldots &= \Zak^{-1} \, U^{-1} \Piref e^{- i \frac{t}{\eps} \hat{h}_{\mathrm{eff}}} \, \Opk^A \bigl (\piref \, u \magW f \magW u^* \, \piref \bigr ) \, e^{+ i \frac{t}{\eps} \hat{h}_{\mathrm{eff}}} \Piref U \, \Zak + \ordern(\eps^{\infty})
		\\
		&= \Zak^{-1} \, U^{-1} \Piref e^{- i \frac{t}{\eps} \hat{h}_{\mathrm{eff}}} \, \Opk^A \bigl (f_{\mathrm{eff}} \bigr ) \, e^{+ i \frac{t}{\eps} \hat{h}_{\mathrm{eff}}} \Piref U \, \Zak + \ordern(\eps^{\infty})
		\\
		&= \Zak^{-1} \, \Pi \, U^{-1} \, \Opk^A \bigl (f \circ T_{\mathrm{eff}} \circ \Phi^{\mathrm{eff}}_t \bigr ) \, U \, \Pi \, \Zak + \ordern(\eps^{2})
		. 
	\end{align*}
	Since two flows are $\ordere{2}$ close if the corresponding hamiltonian vector fields are \cite[Lemma~5.24]{Teufel:adiabaticPerturbationTheory:2003}, we conclude 
	\begin{align*}
		\ldots &= \Zak^{-1} \, \Pi \, U^{-1} \, \Opk^A \bigl (f \circ T_{\mathrm{eff}} \circ \Phi^{\mathrm{eff}}_t \circ T_{\mathrm{eff}}^{-1} \circ T_{\mathrm{eff}} \bigr ) \, U \, \Pi \, \Zak + \ordern(\eps^{2})
		\\
		&= \Zak^{-1} \, \Pi \, U^{-1} \, \Opk^A \bigl (f \circ \Phi^{\mathrm{macro}}_t \circ T_{\mathrm{eff}} \bigr ) \, U \, \Pi \, \Zak + \ordern(\eps^{\infty})
		\\
		&= \Zak^{-1} \, \Pi \, \Opk^A \bigl (f \circ \Phi^{\mathrm{macro}}_t \bigr ) \, \Pi \, \Zak + \ordern(\eps^{2})
		. 
	\end{align*}
	This finishes the proof. 
\end{proof}
%





\bibliographystyle{abbrv}
\bibliography{magWQ}

%

\end{document}